\documentclass[11pt]{amsart}%

\usepackage{amsfonts, amsmath, amsthm, amssymb}

\usepackage[pdftex]{hyperref}
\hypersetup{pdfstartview=FitH,	pdfpagemode=UseNone}

\newcommand{\rem}[1]{}

\newtheorem{theorem}{Theorem}

\newtheorem{claim}[theorem]{Claim}

\newtheorem{conjecture}[theorem]{Conjecture}
\newtheorem{corollary}[theorem]{Corollary}

\newtheorem{lemma}[theorem]{Lemma}

\newcommand{\claimproof}[2]%
{\noindent{\em Proof of Claim \ref{#1}.}
#2\hspace*{\fill}$\Box$~~~~~\vspace{5mm} } 

\newcommand{\term}[1]{\boldsymbol{#1}}

\newcommand{\poly}{\text{poly}}

\newcommand{\ch}{\text{char}}
\newcommand{\rk}{\mathrm{rk}}
\newcommand{\lrsp}{\text{sp}}
\newcommand{\coef}{\mathrm{Coef}}

\newcommand{\ol}{\overline}

\newcommand{\spsp}{\Sigma\Pi\Sigma\Pi}
\newcommand{\ve}{\varepsilon}

\newcommand{\cC}{\mathcal{C}}

\newcommand{\cI}{\mathcal{I}}
\newcommand{\cJ}{\mathcal{J}}
\newcommand{\cL}{\mathcal{L}}
\newcommand{\cM}{\mathcal{M}}

\newcommand{\cS}{\mathcal{S}}

\newcommand{\cV}{\mathcal{V}}

\newcommand{\cR}{\mathcal{R}}

\newcommand{\F}{\mathbb{F}}
\newcommand{\N}{\mathbb{N}}
\newcommand{\C}{\mathbb{C}}

\newcommand{\Z}{\mathbb{Z}}
\newcommand{\D}{\Delta}

\newcommand{\bt}{\mathbf{t}}
\newcommand{\ba}{\mathbf{a}}
\newcommand{\bal}{\boldsymbol{\alpha}}

\newcommand{\rmS}{\mathrm{S}}
\newcommand{\rms}{\mathrm{s}}
\newcommand{\rmbS}{\mathrm{bS}}
\newcommand{\rmbs}{\mathrm{bs}}

\newcommand{\rmH}{\mathrm{H}}

\newcommand{\frs}{\mathfrak{s}}

\newcommand{\VP}{\mathsf{VP}}
\newcommand{\VNP}{\mathsf{VNP}}
\renewcommand{\P}{\mathsf{P}}
\newcommand{\NP}{\mathsf{NP}}

\renewcommand{\k}{\kappa}

\def\({\left(}
\def\){\right)}
\def\<{\langle}
\def\>{\rangle}

\def\ce#1{\left\lceil#1\right\rceil}
\def\le{\leqslant}
\def\ge{\geqslant}

\begin{document}

\title[Set-depth-$\D$ formulas]{Quasi-polynomial Hitting-set for Set-depth-$\D$ Formulas}

\author[Agrawal]{Manindra Agrawal}
\address{Indian Institute of Technology, Kanpur, India.}
\email{manindra@iitk.ac.in}

\author[Saha]{Chandan Saha}
\address{Max Planck Institut f\"ur Informatik, Saarbr\"{u}cken, Germany.}
\email{csaha@mpi-inf.mpg.de}

\author[Saxena]{Nitin Saxena}
\address{Hausdorff Center for Mathematics, Bonn, Germany.}
\email{ns@hcm.uni-bonn.de}




\begin{abstract}
We call a depth-$4$ formula $C$ {\em set-depth-$4$} if there exists a (unknown) partition $X_1\sqcup\cdots\sqcup X_d$ of the variable indices $[n]$ that the top product layer respects, i.e.~$C(\term{x})=\sum_{i=1}^k \prod_{j=1}^{d}$ $f_{i,j}(\term{x}_{X_j})$, where $f_{i,j}$ is a {\em sparse} polynomial in $\F[\term{x}_{X_j}]$. Extending this definition to any depth - we call a depth-$\D$ formula $C$ (consisting of alternating layers of $\Sigma$ and $\Pi$ gates, with a $\Sigma$-gate on top) a \emph{set-depth-$\D$} formula if every $\Pi$-layer in $C$ respects a (unknown) partition on the variables; if $\D$ is even then the product gates of the bottom-most $\Pi$-layer are allowed to compute arbitrary monomials. 

In this work, we give a hitting-set generator for set-depth-$\D$ formulas (over \emph{any} field) with running time polynomial in $\exp((\D^2\log s)^{ \Delta - 1})$, where $s$ is the size bound on the input set-depth-$\D$ formula. In other words, we give a {\em quasi}-polynomial time \emph{blackbox} polynomial identity test for such constant-depth formulas. Previously, the very special case of $\D=3$ (also known as {\em set-multilinear} depth-$3$ circuits) had no known sub-exponential time hitting-set generator. This was declared as an open problem by Shpilka \& Yehudayoff (FnT-TCS 2010); the model being first studied by Nisan \& Wigderson (FOCS 1995). Our work settles this question, not only for depth-$3$ but, up to depth $\epsilon\log s / \log \log s$, for a fixed constant $\epsilon < 1$. 

The technique is to investigate depth-$\D$ formulas via depth-$(\D-1)$ formulas over a {\em Hadamard algebra}, after applying a `shift' on the variables. We propose a new algebraic conjecture about the \emph{low-support rank-concentration} in the latter formulas, and manage to prove it in the case of set-depth-$\D$ formulas. 
\end{abstract}

\maketitle

\section{Introduction} \label{sec:intro}

Polynomial identity testing (PIT) - the algorithmic question of examining if a given arithmetic circuit computes an identically zero polynomial - has received some attention in the recent times, primarily due to its close connection to circuit lower bounds. It is now known that a complete (blackbox) derandomization of PIT for depth-$4$ formulas, via a particular kind of pseudorandom generators, implies $\VP \neq \VNP$ (an algebraic analogue of the much coveted result: $\P \neq \NP$). It is also known that $\VP \neq \VNP$, which amounts to proving exponential circuit lower bounds, must necessarily be shown before proving $\P \neq \NP$ (\cite{V79, SV85}). Blackbox identity testing (equivalently, the problem of designing hitting-set generators), being a promising approach to proving lower bounds, naturally calls for a closer examination. Towards this, some progress has been made in the form of polynomial time hitting set generators for the following models:
\begin{itemize}
 \item depth-$2$ formulas \cite{KS01},
\item depth-$3$ formulas with bounded top fanin \cite{ASSS12, SS11-2},
\item depth-$4$ (bounded depth) constant-occur formulas \cite{ASSS12},
\end{itemize}
and a quasi-polynomial time hitting-set generator for 
\begin{itemize}
 \item multilinear constant-read formulas \cite{AvMV11},
\end{itemize}
among some others (refer to the surveys \cite{SY10, S09, AS09}). The hope is, by studying these special but interesting models we might develop a deeper understanding of the nature of hitting sets and thereby get a clue as to what techniques can be lifted to solve PIT in general (i.e. for depth-$4$ formulas). One such potentially effective technique is the study of \emph{partial derivatives} of formulas.
 
Despite the apparent difference between the approaches of \cite{ASSS12} and \cite{AvMV11}, at a finer level they share a common ingredient - the use of partial derivatives. The partial derivative based method was introduced in the seminal paper by Nisan and Wigderson \cite{NW97} for proving circuit lower bounds, and since then it has been successfully applied (with more sophistications) to prove various interesting results on lower bounds, identity testing and reconstruction of circuits \cite{ASSS12, AvMV11, GKQ12, GKKS12} (refer to the surveys \cite{SY10, CKW11} for much more).
\vspace{0.06in}

\noindent \textbf{Partial derivatives \& shifting - the intuition:} In a way, partial derivatives \emph{shift} the variables by some amount - for e.g., if $f(x_1, x_2, \ldots, x_n)$ is a multilinear polynomial then its partial derivative with respect to $x_1$ is $f(x_1 + 1, x_2, \ldots, x_n) - f(x_1, \ldots, x_n)$. Out of curiosity, one might ask what happens if we shift the polynomial by arbitrary field constants? If we shift a monomial $f(\mathbf{x}) = x_1x_2 \ldots x_n$ by $\mathbf{c} = (c_1, \ldots, c_n) \in \F^n$, $c_i \neq 0$, we get the polynomial $f(\mathbf{x} + \mathbf{c}) = (x_1 + c_1)(x_2 + c_2)\ldots(x_n + c_n)$. Something interesting has happened here: The polynomial $f(\mathbf{x} + \mathbf{c})$ has many \emph{low-support} monomials. By a low-support monomial, we mean that the number of variables involved in the monomial is less than a predefined small quantity, say $\ell$.

Is it possible that shifting has a similar effect on a more general polynomial $f(\mathbf{x})$, i.e. $f(\mathbf{x} + \mathbf{c})$ has low-support monomials with nonzero coefficients, if $f \neq 0$? Surely, this is true if $\mathbf{c}$ is chosen randomly from $\F^n$ (by Schwartz-Zippel \cite{Sch80,Z79}). But, $f$ is not just any arbitrary polynomial, it is a polynomial computed by a formula (say, depth-$3$ or depth-$4$ formula). This makes it an interesting proposition to investigate the following derandomization question: Let $f \neq 0$ be a polynomial computed by a formula. Is it possible to efficiently compute a \emph{small} collection of points $\mathcal{T} \subset \F^n$, such that there exists a $\mathbf{c} \in \mathcal{T}$ for which $f(\mathbf{x} + \mathbf{c})$ has a low-support monomial with nonzero coefficient? 

If the answer to the above question is yes, then it is fairly straightforward to do an efficient blackbox identity test on $f$: For the right choice of $\mathbf{c} \in \mathcal{T}$, $g(\mathbf{x}) = f(\mathbf{x} + \mathbf{c}) \neq 0$ has a low-support monomial. To witness that $g(\mathbf{x}) \neq 0$, it suffices to keep a set of $\ell$ variables intact and set the remaining $n-\ell$ variables to zero in $g$; running over all possible choices of $\ell$ variables whom we choose to keep intact, we can witness the fact that $g \neq 0$. Since $\ell$ is presumably small, $g(\mathbf{x})$ restricted to $\ell$ variables is a sparse polynomial which can be efficiently tested for nonzeroness in a blackbox fashion \cite{KS01}. 

Indeed, we prove that the above intuition is true for the class of set-depth-$\Delta$ formulas (precisely defined in Section \ref{sec:results}) - a highly interesting class capturing many other previously studied models (see Section \ref{sec:results}), including \emph{set-multilinear} depth-$3$ circuits. 
\vspace{0.06in}

\noindent \textbf{Set-multilinear depth-$3$ circuits:} A circuit $C = \sum_{i=1}^{k}{\prod_{j=1}^{d}{f_{i, j}(\term{x}_{X_j})}}$ is called a set-multilinear depth-$3$ circuit if $X_1 \sqcup \ldots \sqcup X_d$ is a partition of the variable indices $[n]$ and $f_{i,j}(\term{x}_{X_j})$ is a linear polynomial in the variables $\term{x}_{X_j}$ i.e. the set of variables corresponding to the partition $X_j$. The set-multilinear depth-$3$ model, first defined by \cite{NW97}, kicked off a flurry of activity. Though innocent-looking, it has led researchers to various arithmetic inventions -- the {\em partial derivative} method for circuit lower bounds \cite{NW97}, noncommutative whitebox PIT \cite{RS05}, the relationship between {\em tensor-rank} and super-polynomial circuit lower bounds \cite{Raz10}, hitting-set for tensors, low-rank recovery of matrices, rank-metric codes \cite{FS12}, and reconstruction (or learnability) of circuits \cite{KS06}. Although, an exponential lower bound for set-multilinear depth-$3$ circuits is known \cite{NW97, RY09}, the closely associated problem of efficient blackbox identity testing on this model remained an open question, until this work.
\vspace{0.06in}

\noindent \textbf{Our contribution: Hitting set for set-depth-$\Delta$ formulas -} A whitebox deterministic polynomial time identity test for set-depth-$\D$ follows from the noncommutative PIT results \cite{RS05}. We are interested in {\em blackbox} PIT and, naturally, we cannot see inside $C$ and the underlying partitions of $[n]$. The only information we have is the circuit-size bound, $s$. To our knowledge, there was no sub-exponential time hitting-set known for the set-depth-$\Delta$ model. Our work improves this situation to quasi-polynomial for \emph{any} underlying field (refer Theorem \ref{thm:set-height-H}). We remark that even the very special case of set-multilinear depth-$3$ circuits had no sub-exponential hitting-set known (see \cite[Problem 27]{SY10}); closest being the recent result of \cite{FS12} where they give a quasi-polynomial hitting-set for {\em tensors}, i.e.~the {\em knowledge of the sets} $X_1,\ldots,X_d$ is required.

Furthermore, set-depth-$4$ covers other well-studied models - diagonal circuits \cite{S08} \& semi-diagonal circuits \cite{SSS12} - that had whitebox identity tests but no blackbox sub-exponential PIT were known. For these (and set-multilinear depth-$3$), our hitting-set has time complexity $s^{O(\log s)}$, although, for general set-depth-$4$ it requires $s^{O(\log^2 s)}$.

Depth-$4$ formulas being the ultimate frontier for PIT (and lower bounds) \cite{AV08}, one might wonder about the utility of our result on hitting-set for set-depth-$\Delta$ formulas beyond $\Delta = 4$. It turns out that there is an interesting connection: We show that a quasi-polynomial hitting set generator for set-depth-$6$ formulas implies a quasi-polynomial hitting set generator for depth-$3$ formulas of the form $C = \sum_{i=1}^{k}{\prod_{j=1}^{d}{{f_{i, j}(\term{x}_{X_j})}^{e_{i, j}}}}$, where $X_1 \sqcup \ldots \sqcup X_d$ defines a partition on $[n]$ and $f_{i, j}$ are linear polynomials. Since arbitrary powers $e_{i ,j} \geq 0$ are allowed, the above depth-$3$ model is stronger than set-multilinear depth-$3$ formulas (as there is no restriction of multilinearity). This appears to be temptingly close to the general depth-$3$ model modulo the partition on variables, and provides us with a good motivation to understand the strength of our approach against depth-$3$ formulas.
\vspace{0.06in}

\noindent {\bf Technical novelty of our approach -} As mentioned before, many works have looked at the partial derivatives of a formula and related matrices, e.g. ~the Jacobian \cite{ASSS12, BMS11}. From a geometric viewpoint, the study via derivatives shifts the variables by an {\em infinitesimal} amount and hopes to discover interesting structure. We take a more radical approach; we shift the circuit by {\em formal} variables and look at how the circuit changes by considering a {\em transfer} matrix $T$. The transfer matrix originates from the study of a formula with field coefficients via a simpler one having {\em Hadamard algebra} coefficients. This makes the transfer process more amenable to an attack using matrices and linear algebra; proving properties that are vaguely reminiscent of the case of top-fanin $k=1$.

The main technicality lies in proving the invertibility of a transfer matrix, which is an exponential-sized matrix. Some of the arguments here are combinatorial in nature involving greedy and binary-search paradigms. 



Although, Hadamard algebra is implicit in the whitebox identity test of \cite{RS05} and the study of PIT over commutative algebras of \cite{SSS09} (Theorem $6$ in \cite{SSS09}), the novelty of our approach lies in understanding the effect of shift by viewing it through the lens of Hadamard algebra, and thereby observing the remarkable phenomenon of \emph{low-support rank concentration}, which in turn implies that a low-support monomial survives after shifting.

We state our results more precisely now.
\subsection{Our results} \label{sec:results}

\textit{Set-depth \& set-height formulas} - Let $C$ be an arithmetic formula over a field $\F$ in $n$ variables $\term{x}$, consisting of alternating layers of addition ($\Sigma$) and multiplication ($\Pi$) gates, with a $\Sigma$-gate on top. The number of layers of $\Pi$-gates in $C$ is called the \emph{product-depth} (or simply \emph{height}) of $C$ and will be denoted by $H$. Naturally, the depth of $C$ - which is the number of layers of gates in $C$ - is either $\D = 2H$ or $2H+1$. Counting the $\Pi$-layers from the top, we label these layers by numbers in the range $[H]$ and will be referring to a layer as the $h$-th $\Pi$-layer in $C$, for $h \in [H]$. 

We say that $C$ is a \emph{set-depth-$\D$} formula if for every $h$-th $\Pi$-layer in $C$, there exists a partition $X_{h,1}\sqcup\cdots\sqcup X_{h,d_h}$ of variable indices $[n]$ that the product gates of the $h$-th $\Pi$-layer respect. In other words, for every $h \in [H]$ the $i$-th product gate in the $h$-th $\Pi$-layer computes a polynomial of the form $\prod_{j=1}^{d_h}{f_{i, j}(\term{x}_{X_{h,j}})}$, where each $f_{i, j}(\term{x}_{X_{h,j}})$ is a set-depth-$(\D-2h)$ formula of height $H-h$ on the variable set $\term{x}_{X_{h,j}}$. If $\D=2H$ then the product gates of the $H$-th $\Pi$-layer are allowed to compute arbitrary monomials, i.e.~here the $H$-th $\Pi$-layer \emph{need not} respect any partition of the variables. 

We will also refer to $C$ as a \emph{set-height-$H$} formula. \emph{Size} of $C$, denoted by $s$ or $|C|$, is the number of gates (including the input gates) in $C$. 

\begin{theorem}[Main] \label{thm:set-height-H}
There is a hitting-set generator for set-height-$H$ formulas, of size $s$, that runs in time polynomial in $\exp((2H^2\log s)^{H+1})$, over any field $\F$. 
\end{theorem}
\noindent {\em Remarks.} 
1. For blackbox PIT of set-multilinear depth-$3$ formulas this gives a {\em quasi-polynomial} time complexity of $s^{O(\log s)}$ - this is the first sub-exponential time 
algorithm.\\
2. For constants $H>1$ the formula may {\em not} be multilinear, though the hitting-set remains quasi-polynomial. The time complexity remains {\em sub-exponential} up to $H = \epsilon\log s / \log\log s$, for a fixed constant $\epsilon < 1$ .
\vspace{0.07in}

An interesting model that is not set-depth-$\D$ but still Theorem \ref{thm:set-height-H} could be applied is - semi-diagonal formula. The reason being the {\em duality} transformation \cite{S08,SSS12} that helps us view it as a  set-depth-$4$ formula. We recall - a depth-$4$ ($\spsp$) formula $C$ is {\em semi-diagonal} if, for all $i$, its $i$-th (top) product-gate computes a polynomial of the form $m_i\cdot\prod_{j=1}^{b}{f_{i,j}^{e_{i,j}}}$, where $m_i$ is a monomial, $f_{i,j}$ is a sum of univariate polynomials, and $b$ is a constant. We give two applications, with similar proofs but, for different looking formulas. 

\begin{corollary}[Semi-diagonal depth-$4$] \label{cor:semi-diagonal}
There is a hitting-set generator for semi-diagonal depth-$4$ formulas, of size $s$, that runs in time $s^{O(\log s)}$ (assuming $\ch(\F)$ zero or large).
\end{corollary}

\begin{corollary}[Set-depth-$3$ with powers] \label{cor:power-set-multilinear}
Consider a depth-$3$ formula $C = \sum_{i=1}^k \prod_{j=1}^{d}$ $f_{i,j}(\term{x}_{X_j})^{e_{i,j}}$, where $f_{i,j}$ is a linear polynomial in $\F[\term{x}_{X_j}]$, $e_{i,j} \in \N$, and $X_1\sqcup\cdots\sqcup X_d$ partitions $[n]$. There is a hitting-set generator for such formulas, of size $s$, that runs in time $s^{O(\log^2 s)}$ (assuming $\ch(\F)$ zero or large). The result continues to hold even if $f_{i,j}$ is a sum of univariates.
\end{corollary} 

\noindent {\em Remarks -} The restriction on $\ch(\F)$ in the above two corollaries comes from the use of the duality trick. We think this restriction can be lifted by using \emph{Galois rings} (\cite{S08, SSS12}), and defining \emph{rank} for a Hadamard algebra over a Galois ring appropriately. We avoid working out the details here just to keep the focus on the main contributions of this work.


\subsection{Organization}

We develop an extensive terminology in Section \ref{sec:basics}, which would be useful later. This section also shows the proof idea at work for the example case of diagonal circuits. 
Section \ref{sec:low-blk-supp} proves the first structural property - a small shift ensures \emph{low-block-support rank-concentration} in a product of polynomials, that have disjoint variables and only low-weight monomials. Starting with this as a base case, Section \ref{sec:low-supp} proves the second structural property - a small shift ensures low-support rank-concentration in set-depth-$\D$ formulas (thus, achieving the presence of a low-support monomial). Finally, the proofs of our main results (or hitting-sets) are completed in Section \ref{sec:hitting-set}.

\section{The basics}\label{sec:basics}

\subsection{Polynomials}

Let $\N:=\Z_{\ge0}$ and $[n]:=\{1,\ldots,n\}$. Let $R$ be a commutative ring. In the motivating cases $R$ will be a field $\F$, which we implicitly assume to be large enough. This we can do as the required field extensions are constructible in deterministic polynomial time \cite{AL86}, further, as in blackbox PIT we are allowed to evaluate the circuit over any `small' field extension. 

Not always will we use bold notation for a vector, hopefully the context will avoid the confusion. For a vector $e\in\Z^n$ we define $|e|:=\sum_i e_i$. Also, let the {\em support} be $\rmS(e):=\{i\,|\,e_i\ne0\}$ and the {\em weight} $\rms(e)$ be its size. For an {\em exponent vector} $e\in\N^n$, we define a {\em coefficient} operator $\coef(e):R[\term{x}]\rightarrow R$ that on a polynomial $f \in R[\term{x}]$ equals the coefficient of $x^e$ in $f$. Clearly, it is an $R$-module homomorphism but is not multiplicative. Define the {\em support of $f$} as $\rmS(f):=\{e\in\N^n \,|\, \coef(e)(f)\ne0 \}$ and the {\em sparsity} $\rms(f)$ be its size. The {\em monomial-weight of $f$} is $\mu(f):=\max_{e\in\rmS(f)}\rms(e)$. Further, define the {\em cone of $f$} as $\cS(f):=\{e'\in\N^n \,|\, \exists e\in\rmS(f), e'\le e \}$, where the inequality is coordinate-wise, and its size as $\frs(f)$. Note that for a sparse polynomial $f$, $\rms(f)$ is small but $\frs(f)$ is usually exponential. 

\begin{lemma}[Cone]\label{lem-cone}
For an $n$-variate polynomial $f$, of degree bound $d$ and monomial-weight $\mu$, we have $\frs(f) \le {n+1 \choose \mu}\cdot{d+\mu \choose \mu}$.
\end{lemma}

For $u,v,a\in\N^n$ define $v!:=\prod_{i\in[n]} v_i!$, ${v \choose u} := \prod_{i\in[n]} {v_i \choose u_i} = \frac{v!}{u!\cdot (v-u)!}$, and $a^{v - u} := \prod_{i \in [n]}{a_i^{v_i - u_i}}$. We keep in mind the conventions: For all $a<b\in\N$, ${a\choose b} = 0$ and ${a\choose 0} = 1$.

\begin{lemma}[Shift on monomials]\label{lem-binom-coeff}
Let $u,v\in\N^n$, $a_1,\ldots,a_n\in R$ and $f=\prod_{i\in[n]}(x_i+a_i)^{v_i}$. Then, 
$\coef(u)(f) = {v \choose u}\cdot a^{v-u}.$
\end{lemma}

For a polynomial $f$ a shift does not change $\mu(f)$ but, might blow up $\rms(f)$ exponentially.

\subsection{Hadamard algebras}

For a commutative ring $R$ and $\k \in \N$, we define the {\em Hadamard algebra} $\rmH_{\k}(R):=(R^{\k},+,\star)$, on the free $R$-module $R^{\k}$, by defining:
$u\star v := (u_i\cdot v_i)_{i\in[\k]}$, where $\cdot$ is the multiplication in $R$. $\rmH_{\k}(R)$ is an $R$-algebra 
(it is closed, associative, distributive and commutative) with the zero vector as {\em zero} and the all-one vector as 
{\em unity}. 

We can now naturally define the {\em polynomial ring over $\rmH_{\k}(R)$}, $\rmH_{\k}(R)[\term{x}]$. It inherits the operations $+,\star$, and all the elements of $\rmH_{\k}(R)$. Also there is an obvious isomorphism between the algebras
$\rmH_{\k}(R)[\term{x}]$ and $\rmH_{\k}(R[\term{x}])$. (View the elements of $\rmH_{\k}(R)$ and $\rmH_{\k}(R)[\term{x}]$ as `column vectors' with entries from $R$ and $R[\term{x}]$, respectively.)

For an $e\in\N^n$ and $f\in \rmH_{\k}(R)[\term{x}]$, we have the natural notions -- {\em coefficient} operator $\coef(e):\rmH_{\k}(R)[\term{x}] \rightarrow \rmH_{\k}(R)$, {\em support} $\rmS(f)\subset\N^n$, and  {\em sparsity} $\rms(f)$. 

\smallskip\noindent 
\textit{Low-support coefficient-space} - For any polynomial $f$ over a Hadamard algebra $\rmH_{\k}(R)$, where $R$ is a field, and $\ell \in \N_{>0}$, define 
$V_{\ell}(f) := \lrsp_{R}\{\coef(e)(f) \,|\, e\in\N^n, \rms(e)<\ell\} \subseteq \rmH_{\k}(R)$.
We call $f$ {\em $\ell$-concentrated over $\rmH_{\k}(R)$} if $V_{\ell}(f) \ =\ \lrsp_{R}\{ \coef(e)(f) \,|\, e\in\N^n\}$.

We can extend the above definition also to the case when $R$ is an integral domain, as we can then work with the associated field of fractions.

We demonstrate the usefulness of Hadamard algebra \& `shifting' in achieving low-support rank concentration, using the example case of diagonal circuits (see Section \ref{sec:appendix_diag}).

\subsection{Proof ideas}
With the spirit of the argument (as in Section \ref{sec:appendix_diag}) in mind, let us state the proof ideas. Let $C(\term{x})=\sum_{i=1}^k \prod_{j=1}^{d}$ $f_{i,j}(\term{x}_{X_j})$, where $f_{i,j}$ is a {\em sparse} polynomial in $\F[\term{x}_{X_j}]$, be a set-depth-$4$ formula. Consider a $\Pi\Sigma\Pi$ formula $$D(\term{x}) := f_1(\term{x}_{X_1}) \star\cdots\star f_d(\term{x}_{X_d})\quad \text{ over } \rmH_k(\F),$$ where the $i$-th coordinate of $f_j(\term{x}_{X_j})$ is $f_{i,j}(\term{x}_{X_j})$. Note that $C(\term{x})$ can be expressed as $(1, \hspace{0.01in} 1, \ldots, 1) \cdot D(\term{x})$, where $\cdot$ is the usual matrix product. Denote $(1, \hspace{0.01in} 1, \ldots, 1)$ by $\mathbf{1}$.

For a subspace $V\subseteq\F^k$ and polynomials $D_1, D_2\in \rmH_k(\F)[\term{x}]$, we say $D_1\equiv D_2\pmod{V}$ if each coefficient of $D_1-D_2$ is in $V$. Somewhat wishfully, we would like to propose a {\em low-support rank-concentration} property:

\begin{conjecture}[Wishful!]\label{conj:0}
If $\ell>\log |D|$ then $D(\term{x})\equiv0 \pmod{V_\ell(D)}$.
\end{conjecture}

If this is true then the coefficient of $\term{x}^e$, in $D$, is in the $\F$-span of those coefficients that correspond to low support, i.e.~$O(\log |D|)$.  Suppose we verify the zeroness of $\pi_S\circ C(\term{x})=\mathbf{1}\cdot D(\pi_S\term{x})$, for $S\in{[n]\choose\ell-1}$ and $\pi_S:x_i\mapsto (x_i \text{ if } i\in S, \text{ else } 0)$. This means that $\forall e\in\N^n$ with $\rms(e)<\ell$ we have $\mathbf{1}\cdot\coef(e)(D)=0$.  Now the conjecture implies that also $\forall e\in\N^n$ with $\rms(e)\ge\ell$ we have $\mathbf{1}\cdot\coef(e)(D)=0$, clearly implying, $C(\term{x})=\mathbf{1}\cdot D(\term{x})=0$. In other words, we have a blackbox PIT for set-depth-$4$ in time $\poly(n^{\log |C|})$.

Unfortunately, Conjecture \ref{conj:0} is easily false! For example, let $D(\term{x})=x_1\cdots x_n$ and $1<\ell\le n$. Then obviously $D(\term{x})\not\equiv0 \pmod{V_\ell(D)}$.

Here is where `shifting' enters the picture. The goal in this paper is to prove that after a `small' shift of the variables, $D$ begins to satisfy something like Conjecture \ref{conj:0}. This requires a rather elaborate study of how a formula changes when shifted; the meat is expressed through certain {\em transfer} equations. Looking ahead, we conjecture (without proof) that the phenomena continue to hold in {\em general} constant-depth formulas.


\subsection{Set-height formulas over Hadamard algebra}
Just as we have defined set-height formulas over a field $\F$ - meaning, the underlying constants come from $\F$, we can also define set-height formula in a natural way over any Hadamard algebra $\rmH_{\kappa}(R)$. The reason we can extend the definition to arbitrary $\rmH_{\kappa}(R)$ is that the defining property of set-height formulas is the existence of a partition of variables for every $\Pi$-layer (irrespective of where the constants of the formula come from). \emph{Size} of a formula $C$ over $\rmH_{\kappa}(R)$ is defined as $\kappa$ times the number of gates in $C$.

Let $C$ be a set-height-$H$ formula (over $\F$) of depth $\D$ - we will count depth of $C$ from the top, i.e.~the top $\Sigma$-gate is at depth $1$.  If $\D$ is even (resp.~odd) then the gates of the bottom-most $\Sigma$-layer compute sparse polynomials (resp.~linear polynomials) in the variables. Let $k$ be the maximum among the fanin of the $\Sigma$-gates of $C$ (barring the gates of the bottom-most $\Sigma$-layer), and $d$ the maximum among the fanin of the $\Pi$-gates in $C$.
\vspace{0.05in}

\noindent \textit{Uniform fanin of $\Sigma$ and $\Pi$-gates} - With the definitions of $k$ and $d$ as above, we can assume that the fanin of every $\Sigma$-gate in $C$ (barring the gates of the bottom-most $\Sigma$-layer) is $k$, and fanin of every $\Pi$-gate is $d$. This can be achieved by introducing `dummy' gates: The `dummy' $\Sigma$-gates introduced as children of a $\Pi$-gate compute the field constant $1$, and the `dummy' $\Pi$-gates introduced as children of a $\Sigma$-gate also compute $1$ except that some of the field constants on the wires are set to zeroes. This process keeps $C$ a set-height-$H$ formula but might bloat up the size from $s$ to $s^{\D}$, although it does not change $k$ and $d$ (according to the way we have defined them). Of course, formula $C$ is not modified physically as it is presented as a blackbox. But the point is, even in the blackbox setting we can treat $C$ as a set-height-$H$ formula with \emph{uniform fanin} of $\Sigma$ and $\Pi$-gates. We will call this uniform fanin of the $\Sigma$ and $\Pi$-gates as the \emph{$\Sigma$-fanin} and \emph{$\Pi$-fanin}, respectively. Note that the definition of $\Sigma$-fanin excludes the gates of the bottom-most $\Sigma$-layer - they are handled next. \\

\vspace{-0.1in}
\noindent \textit{Fanin bound on bottom-most $\Sigma$-gates} - If $\D$ is even, denote the set of monomials computed by the $H$-th $\Pi$-layer by $M$; if $\D$ is odd then $M := \term{x} \cup \{1\}$. The fanin of every gate of the bottom-most $\Sigma$-layer is bounded by $\lambda := |M|+1$. Refer to $\lambda$ as the \emph{sparsity parameter}. 

Henceforth, we will assume uniform $\Sigma$ and $\Pi$-fanin of $C$ ($k$ and $d$ respectively), keeping in mind that the fanin of every gate of the bottom-most $\Sigma$-layer is bounded by $\lambda$. All of $k, d$ and $\lambda$ are in turn bounded by $s$. Denote this class of formulas over $\F$ by $\mathcal{C}_0(k, d, \lambda, \term{x})$.
\vspace{-0.1in}

\noindent \textit{Recursive structure of set-height formulas over Hadamard algebras} -
Let $\mathcal{C}_h(k, d$, $\lambda, \term{x})$ be the class of set-height-($H-h$) formulas, of depth $(\D-2h)$, in the variables $\term{x}$ with $\Sigma$-fanin $k$, $\Pi$-fanin $d$ and sparsity parameter $\lambda$, over the Hadamard algebra $\mathcal{R}_h := \rmH_{k^h}(\F)$. (Eg., to begin with $h = 0$ and the input formula $C \in \mathcal{C}_0(k, d, \lambda, \term{x})$.) Assume that $k, d$ and $\lambda$ are less than $s$, which is the size of the input formula $C$. Let $C_h$ be a formula in $\mathcal{C}_h(k, d, \lambda, \term{x})$.
\begin{equation} \label{eqn:mainformula}
C_h(\term{x}) = \sum_{i \in [k]}{c_i \cdot \prod_{j\in[d]}{f_{i,j}(\term{x}_{X_j})}},
\end{equation}
$c_i \in \mathcal{R}_h$, $f_{i,j}(\term{x}_{X_j})$ is a set-height-$(H-h-1)$ formula over $\mathcal{R}_h$ on the variables $\term{x}_{X_j}$, and $X_1 \sqcup \cdots \sqcup X_d$ is the partition of $[n]$ that the first $\Pi$-layer of $C_h(\term{x})$ respects. Let $\mathcal{R}_{h+1} := \rmH_{k}(\mathcal{R}_h) = \rmH_{k^{h+1}}(\F)$. Define $f_j(\term{x}_{X_j}) := (f_{1,j}(\term{x}_{X_j}),  \ldots, f_{k,j}(\term{x}_{X_j}))^T \in \mathcal{R}_{h+1}[\term{x}_{X_j}]$. Let
\begin{equation*}
D_h(\term{x})\ :=\ f_1(\term{x}_{X_1}) \star\cdots\star f_d(\term{x}_{X_d}) = \prod_{j \in [d]}{f_j(\term{x}_{X_j})} \quad \text{ over } \mathcal{R}_{h+1},
\end{equation*}
where $\star$ denotes the Hadamard product in the algebra $\mathcal{R}_{h+1}$ (extended naturally to the polynomial ring over $\mathcal{R}_{h+1}$). Evidently,
\begin{equation} \label{eqn:hadamard_h}
C_h(\term{x}) = (c_1, \ldots, c_k) \cdot D_h(\term{x}) = \term{c}^T \cdot D_h(\term{x}), 
\end{equation}
where $\cdot$ is the product for matrices over $\mathcal{R}_h[\term{x}]$. We intend to understand the nature of the circuit $C_h(\term{x})$ by studying the properties of the circuit $D_h(\term{x})$ - it is here that the recursive structure reveals itself as in Lemma \ref{lem:rec-struc}. Let $\mathcal{P}_h(h') := \{X_{h',1},\ldots, X_{h',d} \}$ be the partition of $[n]$ that the $h'$-th $\Pi$-layer of $C_{h}$ respects. (Recall that when the depth of $C_h$ is even then the bottom-most $\Pi$-layer need not respect any partition - this attribute would always remain implicit in our discussions.) Define the partition $\mathcal{P}_{h}(h', X_j) := \{ X_{h',1} \cap X_j,\ldots, X_{h',d} \cap X_j \}$ (ignore here the empty sets), for every $1 \leq j \leq d$.
\begin{lemma} \label{lem:rec-struc}
 For every $j \in [d]$, $f_j(\term{x}_{X_j})$ is a set-height-($H-h-1$) formula in $\mathcal{R}_{h+1}[\term{x}_{X_j}]$ with $\Sigma$-fanin $k$, $\Pi$-fanin $d$ and sparsity parameter $\lambda$, i.e.~$f_j(\term{x}_{X_j}) \in \mathcal{C}_{h+1}(k,d, \lambda,  \term{x}_{X_j})$, such that every $h'$-th $\Pi$-layer of $f_j(\term{x}_{X_j})$ respects the partition $\mathcal{P}_h(h'+1, X_j)$. \emph{(Pf. in App. \ref{sec:appendix_basics})} 
\end{lemma}

\subsection{Matrices}

A matrix $M$ with coefficients in ring $R$, and the rows (resp.~columns) indexed by $\cI$ (resp.~$\cJ$) is compactly denoted as: $M\in(\cI\times\cJ\rightarrow R)$. $M$ is simultaneously a map from $\cI\times\cJ$ to $R$, and  a {\em $R$-linear} transformation from $R^{|\cJ|}$ to $R^{|\cI|}$. When $R$ is an integral domain, we denote the {\em rank} by $\rk_R M$. Note that the row-rank and column-rank are equal for a matrix. We call a matrix $M\in(\cI\times\cJ\rightarrow R)$, $|\cI|=|\cJ|-1$, {\em strongly full} if for all $u\in\cJ$, $M_{\cI,\cJ\setminus\{u\}}$ is invertible. For two matrices $M_1, M_2$ and a $R$-module $V$, we write $M_1\equiv M_2 \pmod{V}$ to mean that each column of $M_1-M_2$ is in $V$. For two matrices $M_1\in(\cI_1\times\cJ_1\rightarrow R)$ and $M_2\in(\cI_2\times\cJ_2\rightarrow R)$, the matrices $M_1^{-1}$ (when $|\cI_1|=|\cJ_1|$), $M_1M_2$ (when $\cJ_1=\cI_2$) and $M_1\otimes M_2$ are in $(\cJ_1\times\cI_1\rightarrow R)$, $(\cI_1\times\cJ_2\rightarrow R)$ and $( (\cI_1\times\cI_2)\times(\cJ_1\times\cJ_2) \rightarrow R)$ respectively. For a matrix $M\in R^{\k\times a}$ and an element $v\in \rmH_{\k}(R)$, $v\star M$ is the matrix obtained after taking the Hadamard product of each column with $v$. For two matrices $M_1\in(\cI\times\cJ_1\rightarrow R), M_2\in(\cI\times\cJ_2\rightarrow R)$ the {\em Hadamard-tensor} matrix $M_1\circledast M_2 \in (\cI\times(\cJ_1\times\cJ_2)\rightarrow R)$ is defined as: Its $(j_1,j_2)$-th column is $(M_1)_{\cI, j_1}\star (M_2)_{\cI, j_2}$. We list some intuitive formulas.

\begin{lemma}[Matrices]\label{lem-matrices}
For any column-vector $v$ and matrices $E_i, M_i, Z_i$, with suitable assumptions on the sizes and invertibility, we have:
\begin{enumerate}
\item
$(\otimes_i E_i)\cdot (\otimes_i M_i) = \otimes_i (E_iM_i)$.
\item
$\otimes_i M_i^{-1} = (\otimes_i M_i)^{-1}$.
\item
$(v\star M_1)\cdot M_2 = v\star (M_1M_2)$.
\item
$(Z_1M_1)\circledast(Z_2M_2) = (Z_1\circledast Z_2)\cdot (M_1\otimes M_2)$.
\end{enumerate}
\end{lemma}

\section{Low-block-support rank-concentration} \label{sec:low-blk-supp}

For $i\in[\ell]$, let $f_i \in \rmH_{\k}(\F)[\term{x}_{X_i}]$ be a polynomial of degree at most $\delta$, where the $X_i$'s are disjoint subsets of $[n]$. Define $\mu := \max_i\{\mu(f_i)\}$. By Lemma \ref{lem-cone}, the sparsity parameter $\lambda := \max_i\{\frs(f_i)\}$ of the $f_i$'s is bounded by $(\delta+n+\mu)^{O(\mu)}$. Define $\ell := 2\ce{\log_2 \k}+1$.

Consider the depth-$3$ ($\Pi \Sigma \Pi$) formula over $\rmH_{\k}(\F)$, 
$$D := f_1(\term{x}_{X_1}) \star\cdots\star f_\ell(\term{x}_{X_\ell}) \text{ in } \rmH_{\k}(\F)[\term{x}].$$ 
We {\em shift} it by formal variables $\term{t}$ to get $D(\term{x}+\term{t}) = f_1(\term{x}_{X_1}+\term{t}_{X_1}) \star\cdots\star f_\ell(\term{x}_{X_\ell}+\term{t}_{X_\ell})$ in $\rmH_{\k}(\F[\term{t}])[\term{x}]$. Wlog we can assume that, $\forall i\in[\ell]$, {\em $f_i(\term{t})$ is a unit in $\rmH_{\k}(\F(\term{t}))$}. This is because not being a unit only means that the vector $f_i \in\F(\term{t})^{\k}$ has a zero coordinate, say at place $j\in[\k]$. Then the $j$-th coordinate of $D(\term{t})$ is zero, and we can forget this position altogether; project the setting to the simpler algebra $\rmH_{\k-1}(\F)$.
We {\em normalize} $f_i$ to $f'_i(\term{x}) := f_i(\term{t})^{-1} \star f_i(\term{x}+\term{t})$. Define $D'(\term{x}) := f'_1(\term{x}_{X_1}) \star\cdots\star f'_\ell(\term{x}_{X_\ell})$ in $\rmH_{\k}(\F(\term{t}))[\term{x}]$. 
\begin{equation}\label{eqn-D-prime}
D(\term{x}+\term{t}) \ =\ D(\term{t})\star D'(\term{x}).
\end{equation}

Any exponent $e\in\N^n$, possibly appearing in $D'$, can be written uniquely as $e = \sum_{i\in[\ell]} e_i$, where $e_i\in\cS(f_i)$, because $f_i$'s are on disjoint set of variables. We will frequently use this identification. We define the {\em block-support of $e$}, $\rmbS(e) := \{i\in[\ell] \,|\, e_i \ne0 \}$, and let the {\em block-weight} $\rmbs(e)$ be its size. Based on this we define a relevant vector space, for $l\in\N_{>0}$,
$$\cV_l(D')\ :=\ \lrsp_{\F(\term{t})} \left\{ \coef(e)(D') \,|\, e\in\N^n, \rmbs(e)<l \right\}. $$

\noindent {\em Ordering \& Kronecker-based map} - We define a {\em term ordering} on the monomials $t^e$, $e\in\N^n$, and their inverses. For a $w\in\N^n$ we denote the ordering as $t^e\preceq_w t^{e'}$, or equivalently $1/t^{e'}\preceq_w 1/t^e$, if $\sum_{i\in[n]} w_ie_i \le \sum_{i\in[n]} w_ie'_i$. Note that the ordering is {\em multiplicative} on the monomials, equivalently, the induced ordering on the exponents is {\rm additive}.

For reasons of efficiency, useful later but skippable for now, we assume:
$\prec_w$ keeps the monomials $\left\{\prod_{i\in[\ell]}t^{e_i} \;|\;  \forall i\in[\ell], e_i\in\cS(f_i) \right\}$ distinct. If we fix such a $w\in\N_{>0}$ (note: it could be found in time $\lambda^{O(\ell)}$), then the Kronecker-like homomorphism $\tau : t_i \mapsto y^{w_i}$ ($\forall i\in[n]$) will obviously also map the aforementioned monomials to distinct {\em univariate} ones. We extend $\tau$ to a homomorphism from $\rmH_{\k}(\F[\term{t}])[\term{x}]$ to $\rmH_{\k}(\F[y])[\term{x}]$, by keeping $\term{x}$ unchanged. Its domain can be further extended to a subset of $\rmH_{\k}(\F(\term{t}))[\term{x}]$ (i.e.~as long as $\tau$ does not cause a {\em division by zero}).

We would like to prove something like Conjecture \ref{conj:0} for $D(\term{x}+\term{t})$. Note that it suffices to focus
on $D'(\term{x})$ as its coefficients are all scaled-up by the same nonzero `constant' $D(\term{t})$. The rest of the section is devoted to proving the following theorem. 

\begin{theorem}[Low block-support suffices]\label{thm:low-block-support}
$D'(\term{x}) \equiv0 \pmod{ \cV_{\ell}(D') }$. Further, it remains true under the map $\tau$.
\end{theorem}

\subsection{Shift-\&-normalizing $D$ }
We investigate the effect of shift-\&-normalizing on $f_i$. Write, for $i\in[\ell]$, $f_i(\term{x}_{X_i}) =: \sum_{v_i\in\rmS(f_i)}z_{i,v_i}x^{v_i}$. (Note: $v_i\in\N^n$ and we will denote its $j$-th coordinate by $v_{i,j}\in\N$.) This yields, after shift-\&-normalize (division by {\em units} is allowed in $\rmH_{\k}(\F(\term{t}))$), 
\begin{eqnarray*}
f'_i(\term{x}) := f_i(\term{x}+\term{t})/f_i(\term{t}) =: \sum_{u_i\in\cS(f_i)}z'_{i,u_i}x^{u_i} \quad\in\; \rmH_{\k}(\F(\term{t}_{X_i}))[\term{x}_{X_i}].
\end{eqnarray*}
The last step defines 
\begin{equation}\label{eqn-zi-prime}
z'_{i,u_i} = \coef(u_i)(f'_i) = f_i(\term{t})^{-1}\star  \sum_{v_i\in\rmS(f_i)} z_{i,v_i} {v_i\choose u_i} t^{v_i-u_i}
\end{equation}
for all exponent vectors $u_i\in\rmS(f'_i)\subseteq\cS(f'_i)=\cS(f_i)$. The constant coefficient of $f'_i$, $z'_{i,0} = 1$.

\subsection{Transfer equation of a single polynomial}

Let $f$ be one of the polynomials $f_1,\ldots,f_\ell$ over $\rmH_{\k}(\F)$. Let $S:=S(f)$ and $\cS:=\cS(f)$. For $v\in \cS$ define $z_v:=\coef(v)(f)$, and $z'_v:=\coef(v)(f')$. Since $f$ is a unit, obviously, $S\ne\emptyset$ and $\cS\ne\emptyset$. Let $Z\in([\k]\times \cS\rightarrow\F)$ be such that: Its $v$-th column is the vector $z_v$. Note that exactly $\rms(f)$ of these columns are nonzero. Let $Z'\in([\k]\times\cS\rightarrow\F(\term{t}))$ be such that: Its $u$-th column is the vector $z'_u$. For any $\cC\subseteq\cS(f)$ we define a diagonal matrix $N_\cC\in(\cC\times\cC\rightarrow\F[\term{t}])$ as: Its $u$-th diagonal element is $t^u$. Let the {\em transfer matrix} (of $\Sigma\Pi$ formulas) $T\in(\cS\times\cS\rightarrow\F)$ be such that: Its $(v,u)$-th entry is ${v \choose u}$. We are ready to state the promised {\em transfer equation}.

\begin{lemma}[Transfer equation - primal]\label{lem-struc-eqn}
$ Z' = f(\term{t})^{-1}\star Z N_\cS T N_\cS^{-1}$. \emph{(Pf. in Appendix \ref{sec:appendix_lbs})}
\end{lemma}

For later use, we need a `modulo' version of this transfer equation. As shorthand denote $Z'_{[\k],\cC}$ by $Z'_{\cC}$, for any $\cC\subseteq\cS$. Note that the transfer matrix captures a transformation, from $Z$ to $Z'$, which is clearly invertible. Thus, $T$ is an invertible matrix. Define $T':=(T_{\cS,\cS})^{-1} \in (\cS\times \cS\rightarrow\F)$ and $\cS^*:=\cS\setminus\{0\}$. If $\cS^* = \emptyset$ then it only means that $f\in \rmH_{\k}(\F)$, and is invertible. Such an $f$ could be dropped from $D$ right in the beginning. From now on we assume $\cS^* \ne \emptyset$.
We deduce a modulo version now.

\begin{lemma}[Transfer equation - mod]\label{lem-Teqn-mod}
We have $f(\term{t})^{-1}\star Z \equiv Z'_{\cS^*} N_{\cS^*} T'_{\cS^*,\cS} N_\cS^{-1} \pmod{z'_0}$. Further,  $T'_{\cS^*,\cS}$ is strongly full. \emph{(Pf. in Appendix \ref{sec:appendix_lbs})}
\end{lemma}

\subsection{Transfer equation of $D$: Hadamard tensoring}

For two subsets $B_1,B_2\subset\N^n$ we define $B_1+B_2:=\{b_1+b_2 \,|\, b_1\in B_1, b_2\in B_2\}$, where the sum is coordinate-wise.
For $i\in[\ell]$, let $\cS_i := \cS(f_i)$ and $\cS^*_i := \cS_i\setminus\{0\}$. Define $\cS := \sum_{i\in[\ell]} \cS_i$ and $\cS' := \sum_{i\in[\ell]} \cS^*_i$. Note that there is a natural identification between $\cS'$ and $\times_{i\in[\ell]} \cS^*_i$. We will be implicitly using this. For $i\in[\ell]$, define $Z_i\in ([\k]\times \cS_i \rightarrow\F)$ such that: Its $u_i$-th column is the vector $z_{i,u_i} := \coef(u_i)(f_i)$. Let $Z\in ([\k]\times \cS \rightarrow\F)$ such that: Its $u$-th column is the vector $z_u := \coef(u)(D)$. Note that $Z = \circledast_{i\in[\ell]} Z_i$. For $i\in[\ell]$, define $Z'_i\in ([\k]\times \cS^*_i \rightarrow\F)$ such that: Its $v_i$-th column is the vector $z'_{i,v_i} := \coef(v_i)(f'_i)$. (Note that $Z'_i$ has fewer columns than $Z_i$.) Let $Z'\in ([\k]\times \cS' \rightarrow\F)$ such that: Its $v$-th column is the vector $z'_{i,v} := \coef(v)(D')$. Note that $Z' = \circledast_{i\in[\ell]} Z'_i$. For any $\cC\subseteq\cS$ we define a diagonal matrix $N_\cC\in(\cC\times\cC\rightarrow\F[\term{t}])$ as: Its $u$-th diagonal element is $t^u$. For $i\in[\ell]$, define $T'_i := T'_{\cS^*_i,\cS_i}$.

Let the {\em transfer matrix} (of $\Pi\Sigma\Pi$ formulas) $T'\in (\cS'\times \cS \rightarrow\F )$ be $\otimes_{i\in[\ell]}T'_i$.

\begin{lemma}[Tf. eqn. depth-$3$]\label{lem-Teqn-depth3}
$ D(\term{t})^{-1}\star Z \equiv  Z' N_{\cS'} T' N_{\cS}^{-1} \pmod{\cV_{\ell}(D')}$. \emph{(Pf. App. \ref{sec:appendix_lbs})}
\end{lemma}

\subsection{Combinatorial juggernaut: To select columns of $T'$}

Recall that $T'$ has rows (resp.~columns) indexed by $\cS'$ (resp.~$\cS$) and has entries in $\F$. Let $\cM$ be some $\k>0$ columns that we intend to remove from $T'$; we call them {\em marked} and the others $\cS\setminus\cM$ are {\em unmarked}. We make the following claim about the submatrices of $T'$ not involving $\cM$.

\begin{theorem}[Invertible minor]\label{thm-sub-prop}
There exist unmarked columns $\cC\subseteq \cS$, $|\cC|=|\cS'|$, such that $|T'_{\cS',\cC}| \ne0$. \emph{(Proof in Appendix \ref{sec:appendix_lbs})}
\end{theorem}

\subsection{$T'$ on the nullspace of $Z$: Finishing Theorem \ref{thm:low-block-support} }

Recall that the columns of $Z$ are indexed by $\cS$. Think of these {\em ordered} by the weight vector $w$, as discussed in the beginning of this section. Pick a basis $\cM$, size at most $\k$, of the column vectors of $Z$ by {\em starting from the largest column}. Formally, $\cM$ gives the unique (once $\prec$ is fixed) basis such that for each $u$-th, $u\in \cS\setminus\cM$, column of $Z$ there exist columns $u_1,\ldots,u_r\in\cM$ spanning the $u$-th column, and $u\prec u_r\prec\cdots\prec u_1$. We think of the columns $\cM$ of $T'$ {\em marked}, and invoke Theorem \ref{thm-sub-prop} to get the $\cC\subsetneq\cS$. We define an $A\in (\cS\times\cC \rightarrow \F)$: If $a$ is the $v$-th column of $A$ then $Z\cdot a=0$ expresses the $\F$-linear dependence of $z_v$ on $\{z_{v'} \,|\, v'\in\cM, v\prec v'\}$; in particular, the {\em least} row where $a$ is nonzero is the $v$-th, the entry being $1$. Recall the transfer equation, Lemma \ref{lem-Teqn-depth3}, for the following.

\begin{lemma}[$T'$ on nullspace of $Z$]\label{lem-tna-det}
$|T' N_\cS^{-1} A|\ne0$. Further, the leading nonzero inverse-monomial in the determinant has the coefficient $|T'_{\cS',\cC}|$. \emph{(Proof in Appendix \ref{sec:appendix_lbs})}
\end{lemma}
Finally, we use $A$ to finish the proof of our main structure theorem.

\begin{proof}[Proof of Theorem \ref{thm:low-block-support}]
From the transfer equation, Lemma \ref{lem-Teqn-depth3}, we recall 
$$D(\term{t})^{-1}\star Z \equiv  Z' N_{\cS'} T' N_{\cS}^{-1} \pmod{\cV_{\ell}(D')}.$$
Right-multiplying by $A$, we get
\begin{equation}\label{eqn-rta-inv}
0\ =\ D(\term{t})^{-1} \star (ZA) \equiv Z'  N_{\cS'} T' N_{\cS}^{-1} A \pmod{\cV_\ell(D')}.
\end{equation}
Since $T' N_{\cS}^{-1} A$ is invertible from Lemma \ref{lem-tna-det} and $N_{\cS'}$ is obviously invertible, we get 
$$Z' \equiv 0 \pmod{\cV_\ell(D')}.$$
(Here we do use that the matrices are over $\F(\term{t})$ and that $\cV_\ell(D')$ is an $\F(\term{t})$-vector space.) This immediately implies the first part of Theorem \ref{thm:low-block-support}, as $Z'$ collected exactly those coefficients of $D'$ that we a priori did not know in $\cV_\ell(D')$. The second part of the theorem follows easily as: (1) $\tau$ keeps $D(\term{t})$ a unit, and (2) $\tau$ corresponds to the correct term ordering $\preceq_w$. These two properties allow the above proof also work after applying $\tau$. 
\end{proof}

\section{Low-support rank-concentration} \label{sec:low-supp}

We will prove that a set-height-$H$ formula, after a `small' shift, begins to have `low'-support rank-concentration. The proof is by induction on the height of the formulas over Hadamard algebras. For this, we would need the following concepts. 

For $H > h \in\N$, let $\bt_h := \{t_{H-1}, \ldots, t_{h+1}, t_h\}$ be a set of formal variables and $\F(\bt_h)$ be the function field. These $\bt_h$-variables are different from the variables $\term{x}$ involved in the formula $C$. Let $\cR'_h := \rmH_{k^h}(\F(\bt_h))$ be a Hadamard algebra over $\F(\bt_h)$; $k^h = \dim_{\F(\bt_h)} \cR'_h$. 
Further, $\cR'_{h+1}[t_h]$ denotes the (univariate) polynomial ring over $\cR'_{h+1}$, and $\cR'_{h+1}(t_h)$ is the corresponding {\em ring of fractions}. ($\cR'_{h+1}(t_h)$ is basically $\rmH_{k^{h+1}}(\F(\bt_h))$.)
\vspace{0.05in}

\noindent \textit{Low-support shift for $\mathcal{C}_h(k, d, \lambda, \term{x})$} - Let $\tau_h$ be a map from $\F[\term{x}]$ to $\F(\bt_h)[\term{x}]$ defined as,
$$\tau_h: x_i \mapsto x_i + \alpha_{H-1, i}\, t^{a_{H-1, i}}_{H-1} + \cdots + \alpha_{h, i}\, t^{a_{h,i}}_h, \hspace{0.1in} \text{for $x_i \in \term{x}$},$$
$a_{H-1,i}, \ldots, a_{h,i} \in \Z^+$ and $\alpha_{H-1,i},\ldots,\alpha_{h,i} \in \F$. ($\tau_h$ fixes $\F$, i.e.~$\tau_h(c) = c$ for $c \in \F$.) In short, we will write $\tau_h: \term{x} \mapsto \term{x} + \bal_h\,\bt_h^{\ba_h}$. For $\ell_h \in \N$, the map $\tau_h$ (as above) is called an \emph{$\ell_h$-support shift} for the class of formulas $\mathcal{C}_h(k, d, \lambda, \term{x})$ if for every formula $C_h \in \mathcal{C}_h(k, d, \lambda, \term{x})$, the polynomial $\tau_h(C_h(\term{x})) = C_h(\term{x} + \bal_h\,\bt^{\ba_h}_h)$ is $\ell_h$-concentrated over $\cR'_h$.

For the rest of our discussion, we will fix $\ell_h$ as follows, for $H > h \geq 0$:
$$
\ell_h := 
\begin{cases}
(2H \lceil H\log_2 k \rceil)^{H-h-1} \cdot 2\ce{ H\log_2 (k\lambda) } + 1, & \text{if $\D$ is even,} \\
 (2H \lceil H\log_2 k \rceil)^{H-h} + 1, & \text{if $\D$ is odd (\& for $h=H$, $\ell_H := 2$).} 
\end{cases}
$$
The above setting satisfies the relation $\ell_h = (\ell_{h+1} - 1)H(\ell-1) + 1$, where $\ell := 2\lceil H\log_2 k \rceil + 1$, for every $H-1 > h \geq 0$ (and also for $h=H-1$ when $\D$ is odd). 

Recall Equation \ref{eqn:hadamard_h} that says - for each $h\in\{0,\ldots,H-1\}$ and $C_h$, there exists $\term{c}\in \rmH_k(\cR_h)$ such that $C_h = \term{c}^T\cdot D_h$. This section is dedicated to proving the following theorem.

\begin{theorem}[Low support suffices]\label{thm:low-support}
We can construct $\tau_0$ such that $\tau_0\circ D_0$ is $\ell_0$-concentrated over $\cR'_1[t_0]$, in time polynomial in $(d+n+\ell_0)^{\ell_0}$, where $n := |\term{x}|$. 
\end{theorem}

\noindent \textit{Proof strategy ahead} - The idea is to construct the map $\tau_h$ by applying induction on height $H - h$ of the class $\mathcal{C}_h(k, d, \lambda, \term{x})$. By Equation \ref{eqn:hadamard_h},
\begin{equation*}
C_h(\term{x}) = c^T \cdot (f_1(\term{x}_{X_1}) \star\cdots\star f_d(\term{x}_{X_d})). 
\end{equation*}
From Lemma \ref{lem:rec-struc}, $f_j(\term{x}_{X_j}) \in \mathcal{C}_{h+1}(k,d, \lambda,  \term{x}_{X_j})$. By definition, $\tau_{h+1}:x_i \mapsto x_i + \alpha_{H-1, i}\, t^{a_{H-1, i}}_{H-1} + \cdots + \alpha_{h+1, i}\, t^{a_{h+1,i}}_{h+1}$ is an $\ell_{h+1}$-support shift for $\mathcal{C}_{h+1}(k,d, \lambda,  \term{x}_{X_j})$ for every $1 \leq j \leq d$. Here is where we use induction on height $H-h$: We will build the map $\tau_h$ from the inductive knowledge of $\tau_{h+1}$. Basically, we will show that it is possible to efficiently compute $a_{h,1}, \ldots, a_{h,n} \in \Z^{+}$ and $\alpha_{h,1}, \ldots, \alpha_{h,n} \in \F$ such that $\tau_h: x_i \mapsto \tau_{h+1}(x_i) + \alpha_{h,i}\,t^{a_{h,i}}_h$ is an $\ell_h$-support shift for $\mathcal{C}_h(k, d, \lambda, \term{x})$. \\

\noindent \textit{The proof of Theorem \ref{thm:low-support}.} The proof proceeds by induction on height $H-h$ of the class $\mathcal{C}_h(k, d, \lambda, \term{x})$ (in other words, {\em reverse} induction on $h$). The induction hypothesis is that $\tau_{h+1}$, an $\ell_{h+1}$-support shift for the class $\mathcal{C}_{h+1}(k, d, \lambda, \term{x})$, can be constructed in time polynomial in $(d+n+\ell_{h+1})^{\ell_{h+1}}$, where $n := |\term{x}|$. Overall this means, by varying $h\in[0,...,H-1]$, we get a hitting-set of size polynomial in $\Pi_{h=0}^{H-1}{(d+n+\ell_h)^{\ell_h}} \le (d+n+\ell_0)^{\sum_h\ell_h} < (d+n+\ell_0)^{2\ell_0}$. We discuss the base case and the inductive step in separate detail. Keep in mind that $f_j(\term{x}_{X_j}) \in \mathcal{C}_{h+1}(k,d, \lambda,  \term{x}_{X_j})$.

\subsection{Base case ($h+1\ge H-1$)}

The base case is when $H-h-1 = 1$ or $0$, i.e.~$f_j(\term{x}_{X_j})$'s are sparse polynomials or linear polynomials over $\cR_{h+1}$, depending on whether $\D$ is even or odd, respectively. These two base cases have varying level of difficulty.
If $H-h-1 = 0$ then $\ell_{h+1} = \ell_{H} = 2$, hence taking $\tau_{H}$ as the identity map suffices (since $f_j(\term{x}_{X_j})$'s are linear polynomials) as an $\ell_H$-support shift for the class $\mathcal{C}_H(k, d, \lambda, \term{x})$. If $H-h-1 = 1$ then $f_j(\term{x}_{X_j})$'s are sparse polynomials. We first prove an, independently interesting, property.

\begin{lemma}[Sparse polynomial]\label{lem:sparse-poly}
Let $f \in \rmH_{\k}(\F)[\term{x}]$ be a polynomial with degree bound $\delta$. Let $\ell' := 1 + \min\{2\ce{\log_2 (\k\cdot\rms(f))}, \mu(f)\}$. We can construct a map $\sigma:x_i\mapsto x_i+t^{b_i}$, in time polynomial in $(\delta+n+\ell')^{\ell'}$, such that $\sigma(f)$ is $\ell'$-concentrated over $\rmH_{\k}(\F(t))$. \emph{(Ap. \ref{sec:appendix_ls})}
\end{lemma}
Now we apply the lemma to the sparse polynomial $f_j(\term{x}_{X_j})$, which has the sparsity parameter $\lambda$.
Hence we define $\tau_{h+1} = \tau_{H-1}: x_i \mapsto x_i + t_{H-1}^{b_i}$ (in other words, $a_{H-1,i}:=b_i$). This, by Lemma \ref{lem:sparse-poly}, ensures that the concentration parameter is $2\ce{ \log_2 (k^{H-1}\cdot\lambda) } + 1 \le$ 
$2\ce{ H\log_2 (k\lambda) } + 1$ $= \ell_{H-1}$ $= \ell_{h+1}$. Finally, $\tau_{H-1}$ is an $\ell_{H-1}$-support shift for the class $\mathcal{C}_{H-1}(k, d, \lambda, \term{x})$, and it can be constructed in time polynomial in $(d+n+\ell_{H-1})^{\ell_{H-1}}$.
\vspace{-0.05in}

\subsection{Induction ($h+1$ to $h$)}

Let $\widehat{f}_j(\term{x}_{X_j}) := \tau_{h+1}(f_j(\term{x}_{X_j}))$. Then, 
\begin{equation*}
\widehat{D}_h(\term{x}):= \tau_{h+1}(D_h(\term{x})) =  \widehat{f}_1(\term{x}_{X_1}) \star\cdots\star \widehat{f}_d(\term{x}_{X_d}),
\end{equation*}
where every $\widehat{f}_j$ is $\ell_{h+1}$-concentrated over $\cR'_{h+1}$ (by induction hypothesis). Let $\bt := \{t_{h,1}, \ldots, t_{h,n} \}$ be a set of `fresh' formal variables. (We will keep in mind that the $\bt$-variables would be eventually set as univariates in a variable $t_h$.) As before in Eqn. \ref{eqn-D-prime},
\begin{equation*}
\widehat{D}_h(\term{x} + \bt) = \prod_{j \in [d]}{\widehat{f}_j(\term{x}_{X_j} + \bt_{X_j})} = \prod_{j \in [d]}{\widehat{f}_j(\bt_{X_j}) \star \widehat{f}'_j(\term{x}_{X_j})} = \widehat{D}_h(\bt) \star \widehat{D}'_h(\term{x}).
\end{equation*}
In the same spirit as Theorem \ref{thm:low-block-support}, we would like to show that 
$\widehat{D}'_h(\term{x})\equiv0 \pmod{\cV_{\ell}(\widehat{D}'_h)}$, 
where $\cV_{\ell}(\widehat{D}'_h)\ :=\ \lrsp_{\F(\bt_{h+1}, \term{t})} \left\{ \coef(e)(\widehat{D}'_h) \,|\, e\in\N^n, \rmbs(e)< \ell \right\}$, and $\ell = 2 \lceil H\log_2 k\rceil + 1$. As before (see `key argument' in Lemma \ref{lem:sparse-poly}), it is sufficient to prove the typical case (i.e.~product of the first $\ell$ polynomials), 
$ \widehat{D}'_{h, \ell}(\term{x}):= \prod_{j \in [\ell]}{\widehat{f}'_j(\term{x}_{X_j})} \equiv0 \pmod{\cV_{\ell}(\widehat{D}'_{h, \ell})}$
Towards this, we define the {\em truncated} polynomials, $\widehat{g}_j(\term{x}_{X_j}):= \sum_{e:s(e) < \ell_{h+1}}{\coef(e)(\widehat{f}_j)\,\term{x}^e_{X_j}}$ and let the corresponding product be
$\widehat{E}_h(\term{x}):= \prod_{j \in [d]}{\widehat{g}_j(\term{x}_{X_j})}$.
Sparsity of $\widehat{g}_j(\term{x}_{X_j})$ over $\cR'_{h+1}$ is bounded by $(d^{H-h-1}+n+\ell_{h+1})^{\ell_{h+1}}$ $=: \lambda_h$. Mimicking the notations on $\widehat{D}_h$ let,
\begin{equation*}
\widehat{E}_h(\term{x} + \bt) = \prod_{j \in [d]}{\widehat{g}_j(\term{x}_{X_j} + \bt_{X_j})} = \widehat{E}_h(\bt) \star \widehat{E}'_h(\term{x}) \hspace{0.1in} \text{and} \hspace{0.1in} \widehat{E}'_{h, \ell}(\term{x}):= \prod_{j \in [\ell]}{\widehat{g}'_j(\term{x}_{X_j})}.
\end{equation*}
By Theorem \ref{thm:low-block-support}, we can find $a_{h, 1}, \ldots a_{h, n} \in \Z^{+}$ in time 
$(d\lambda_h)^{O(\ell)} =$ $(d+n+\ell_{h})^{O(\ell_{h})}$ such that by setting $t_{h, i} = \alpha_{h,i}\, t_h^{a_{h,i}}$ (any $\alpha_{h,i} \in \F\setminus\{0\}$ works), where $t_h$ is a `fresh' formal variable, we can ensure that the following is satisfied: 
\begin{equation} \label{eqn:trunc_l1}
\widehat{E}'_{h, \ell}(\term{x}) \equiv0 \pmod{\cV_{\ell}(\widehat{E}'_{h, \ell})}.
\end{equation} 
The claim is that the same setting $t_{h, i} = \alpha_{h,i}\, t_h^{a_{h,i}}$ (now with carefully chosen $\alpha_{h,i}$'s) also ensures that $\widehat{D}'_{h, \ell}(\term{x}) \equiv0 \pmod{\cV_{\ell}(\widehat{D}'_{h, \ell})}$. Consequently, $\widehat{D}'_h$ is $(\ell-1)(\ell_{h+1}-1)+1 < \ell_h$ concentrated over $\cR_{h+1}'(t_h)$. This is what we argue next. Equation \ref{eqn:trunc_l1} implies
\begin{equation} \label{eqn:trunc_l2}
\widehat{E}_{h, \ell}(\term{x} + \bal\,\bt) = \prod_{j \in [\ell]}{\widehat{g}_j(\term{x}_{X_j} + \bal_{X_j}\,\bt_{X_j})} = \widehat{E}_{h, \ell}(\bal\,\bt) \star \widehat{E}'_{h, \ell}(\term{x}) \equiv0 \pmod{\cV_{\ell}(\widehat{E}_{h, \ell}(\term{x} + \bal\,\bt))},
\end{equation}
where (reusing symbol) $\bt := (t_h^{a_{h,1}}, \ldots, t_h^{a_{h,n}})$ and $\bal := (\alpha_{h,1},\ldots,\alpha_{h,n})$. Define, $\widehat{D}_{h, \ell}(\term{x}):= \prod_{j=1}^{\ell}{\widehat{f}_j(\term{x}_{X_j})}.$ We need to take a closer look at how the coefficients of $\widehat{D}_{h, \ell}(\term{x})$, $\widehat{D}_{h, \ell}(\term{x} + \bal\,\bt)$, $\widehat{E}_{h, \ell}(\term{x})$ and $\widehat{E}_{h, \ell}(\term{x} + \bal\,\bt)$ are related to each other. Towards this, define:
\begin{eqnarray*}
\widehat{z}_{j, u_j} &:=& \coef(u_j)(\widehat{f}_j(\term{x}_{X_j})) \in \cR'_{h+1}, \\
\widehat{z}'_{j, u_j} &:=& \coef(u_j)(\widehat{f}_j(\term{x}_{X_j} + \bal_{X_j}\,\bt_{X_j})) \in \cR'_{h+1}[t_h], \\
\widetilde{z}_{j, u_j} &:=& \coef(u_j)(\widehat{g}_j(\term{x}_{X_j})) \in \cR'_{h+1}; \text{ equals } \widehat{z}_{j, u_j} \text{ if $u_j \in S(\widehat{g}_j)$}, \\
\widetilde{z}'_{j, u_j} &:=& \coef(u_j)(\widehat{g}_j(\term{x}_{X_j} + \bal_{X_j}\bt_{X_j})) \in \cR'_{h+1}[t_{h}].
\end{eqnarray*}
Let,
\begin{eqnarray*}
\widehat{B}_j &:=& \{u_j: \widehat{z}_{j, u_j} \text{ is in the } \F(\bt_{h+1}) \text{-basis of the coefficients of } \widehat{f}_j\} \text{ and } \\
\widetilde{B}_j &:=& \{u_j: \widetilde{z}_{j, u_j} \text{ is in the } \F(\bt_{h+1}) \text{-basis of the coefficients of } \widehat{g}_j\}
\end{eqnarray*}
with respect to some fixed basis that comprises coefficients of monomials of as low support as possible.
Note that $\widehat{B}_j = \widetilde{B}_j =: B_j$, as $\widehat{f}_j$ is $\ell_{h+1}$-concentrated over $\cR'_{h+1}$.
 
The crucial observation is that, for any $v_j \in B_j$, $\widehat{z}'_{j, v_j}$ gets a  $t_h$-free contribution only from the monomial $x^{v_j}$, thus, its basis representation looks like:
$$
\widehat{z}'_{j, v_j} = (1 + a(v_j, v_j)) \cdot \widehat{z}_{j, v_j} + \sum_{u_j \in B_j \backslash\{v_j\} }{a(u_j, v_j) \cdot \widehat{z}_{j, u_j}}, $$ 
where $a$'s are in $\F(\bt_{h+1})[t_h]$ and $t_h$ divides each $a(\cdot,v_j)$. Similarly, 
$$
\widetilde{z}'_{j, v_j} = (1 + b(v_j, v_j)) \cdot \widehat{z}_{j, v_j} + \sum_{u_j \in B_j \backslash\{v_j\} }{b(u_j, v_j) \cdot \widehat{z}_{j, u_j}}, $$
where $b$'s are in $\F(\bt_{h+1})[t_h]$ and $t_h$ divides each $b(\cdot,v_j)$.
Now define the following matrices:
\begin{eqnarray*}
\widehat{Z}_{j} \in ([k^{h+1}] \times B_j \rightarrow \F(\bt_{h+1})) &;& \text{with $u_j$-th column $\widehat{z}_{j, u_j}$}, \\
\widehat{Z}'_{j} \in ([k^{h+1}] \times B_j \rightarrow \F(\bt_{h})) &;& \text{with $u_j$-th column $\widehat{z}'_{j, u_j}$}, \\
\widetilde{Z}'_{j} \in ([k^{h+1}] \times B_j \rightarrow \F(\bt_{h})) &;& \text{with $u_j$-th column $\widetilde{z}'_{j, u_j}$}.
\end{eqnarray*} 
From the above crucial observation, 
\begin{equation} \label{eqn:crucial_obs1}
\widehat{Z}'_{j} = \widehat{Z}_{j} \cdot \widehat{M}' \hspace{0.1in} \text{ and } \hspace{0.1in} \widetilde{Z}'_{j} = \widehat{Z}_{j} \cdot \widetilde{M}',
\end{equation}
where $\widehat{M}', \widetilde{M}' \in (B_j \times B_j \rightarrow \F(\bt_{h+1})[t_h])$ with rows indexed by $u_j \in B_j$ and columns indexed by $v_j \in B_j$. The $(u_j, v_j)$-th entry of $\widehat{M}'$ contains $a(u_j, v_j)$ if $u_j \neq v_j$, otherwise $1 + a(u_j, v_j)$ if $u_j = v_j$. Similarly, the $(u_j, v_j)$-th entry of $\widetilde{M}'$ contains $b(u_j, v_j)$ if $u_j \neq v_j$, otherwise $1 + b(u_j, v_j)$ if $u_j = v_j$. Note that both $\widehat{M}'$ and $\widetilde{M}'$ are invertible over $\F(\bt_{h+1})(t_h)$ as $\det(\widehat{M}') \equiv \det(\widetilde{M}') \equiv 1 \pmod{t_h}$. Therefore,
\begin{equation} \label{eqn:crucial_obs2}
\widehat{Z}'_{j} = \widetilde{Z}'_{j} \cdot (\widetilde{M}'^{-1} \widehat{M}') \hspace{0.1in} \text{ and } \hspace{0.1in } \widetilde{Z}'_{j} = \widehat{Z}'_{j} \cdot (\widetilde{M}'^{-1} \widehat{M}')^{-1}.
\end{equation}

Now observe that any coefficient of $\widehat{D}_{h,\ell}(\term{x} + \bal\,\bt)$ is an $\F(\bt_h)$-linear combination of the columns of $\circledast_{j \in [\ell]}{\widehat{Z}_{j}}$ (by the definition of $B_j$), which by Equation \ref{eqn:crucial_obs1} (\& Lemma \ref{lem-matrices}-(4)) is an $\F(\bt_h)$-linear combination of the columns of $\circledast_{j \in [\ell]}{\widehat{Z}'_{j}}$ - this in turn is an $\F(\bt_h)$-linear combination of the columns of $\circledast_{j \in [\ell]}{\widetilde{Z}'_{j}}$ (by Equation \ref{eqn:crucial_obs2}). By Equation \ref{eqn:trunc_l2}, any $\F(\bt_h)$-linear
combination of the columns of $\circledast_{j \in [\ell]}{\widetilde{Z}'_{j}}$ can be expressed as an $\F(\bt_h)$-linear
combination of those columns $u$ of $\circledast_{j \in [\ell]}{\widetilde{Z}'_{j}}$ for which $\rmbs{(u)} < \ell$, which in turn can be expressed as an $\F(\bt_h)$-linear combination of those columns $u$ of $\circledast_{j \in [\ell]}{\widehat{Z}'_{j}}$ for which $\rmbs{(u)} < \ell$ (by Equation \ref{eqn:crucial_obs2} again). In other words, we have shown the following:
$\widehat{D}_{h, \ell}(\term{x} + \bal\,\bt) \equiv0 \pmod{\cV_{\ell}(\widehat{D}_{h, \ell}(\term{x} + \bal\,\bt))}$.
This would imply that
$\widehat{D}'_{h, \ell}(\term{x}) \equiv0 \pmod{\cV_{\ell}(\widehat{D}'_{h, \ell})}$,
if we choose $\bal$ so that the map $t_{h,i} \mapsto \alpha_{h,i}\,t_h^{a_{h,i}}$ ensures that $\widehat{f}_j(\bal_{X_j}\,\bt_{X_j})^{-1}$ is \emph{well-defined} in $\cR'_{h+1}(t_h)$. Such an $\bal$ can be constructed, by Lemma \ref{lem:hitting-set-non-sparse-poly}, in time polynomial in $\lambda_h =$ $(d^{H-h-1}+n+\ell_{h+1})^{\ell_{h+1}}$.
Therefore, $\tau_h: x_i \mapsto \tau_{h+1}(x_i) + \alpha_{h,i}\,t_h^{a_{h,i}}$ is such that $\tau_h(D_h(\term{x}))$ is $\ell_h$-concentrated over $\cR'_{h+1}[t_h]$. Since $C_h(\term{x}) = c^T \cdot D_h(\term{x})$, hence $\tau_h(C_h(\term{x}))$ is $\ell_h$-concentrated over $\cR'_h$. 
This finishes the construction of $\tau_h$, given $\tau_{h+1}$, in time $(d+n+\ell_{h})^{O(\ell_{h})}$. \qed



\begin{lemma}[Preserve invertibility]\label{lem:hitting-set-non-sparse-poly}
Let $f \in \rmH_{\k}(\F)[\term{x}]$ be a polynomial with degree bound $\delta$. Assume that $f$ is $\ell'$-concentrated over $\rmH_{\k}(\F)$, and that $f^{-1} \in \rmH_{\k}(\F(\term{x}))$. 
Then, we can contruct an $\bal\in\F^n$, in time polynomial in $\k (\delta+n+\ell')^{\ell'}$, such that $f(\bal)^{-1} \in \rmH_{\k}(\F)$.
\end{lemma}
\noindent (Proof in Appendix \ref{sec:appendix_ls}.)

\section{Reading off the hitting-set} \label{sec:hitting-set}

\subsection{Proof of Theorem \ref{thm:set-height-H} }

Suppose we are given a blackbox access to a set-height-$H$ nonzero formula $C$ of size $s$, more so we can think of $C =C_0 \in \mathcal{C}_0(k, d, \lambda, \term{x})$. Using Theorem \ref{thm:low-support} we can construct a map $\tau_0:\F[\term{x}]\mapsto \F[\term{t}_0][\term{x}]$ such that $\widehat{D} := \tau_0\circ D_0$ is $\ell_0$-concentrated over $\cR'_1[t_0]$, in time $(d+n+\ell_0)^{O(\ell_0)}$. Clearly, $\widehat{D} \in 
\rmH_k(\F[\bold{t}_0])[\bold{x}]$ and $C' := \tau_0\circ C = \bold{c}^T\cdot\widehat{D}$. 
For $X\subseteq[n]$ of size at most $\ell_0$, define $\sigma_{X}:x_j\mapsto (x_j \text{ if } j\in X, \text{ else }0)$ for all $j\in[n]$. Clearly, $\sigma_X\circ C'$ is only $\ell_0$-variate, thus it has sparsity $(d^H+\ell_0)^{O(\ell_0)}$. By the assumption on $\widehat{D}$ we know that there exists such an $X$ for which $\sigma_{X}\circ C' \ne0$. Thus, using standard sparse PIT methods (see \cite{BHLV09}) we can construct a hitting-set for $C'$, in time $(d^H+n+\ell_0)^{O(\ell_0)} = 2^{O( \ell_0 H \log (s+\ell_0) )} = \exp(O( \ell_0 H^2 \log s )) $, which is time polynomial in $\exp((2H^2\log s)^{H+1})$. \qed

\subsection{Proof of Corollary \ref{cor:semi-diagonal} }

Suppose we are given a blackbox access to a semi-diagonal formula $C = \sum_{i=1}^k m_i\cdot\prod_{j=1}^{b}{f_{i,j}^{e_{i,j}}}$ over field $\F$, where $m_i$ is a monomial, $f_{i,j}$ is a sum of univariate polynomials, and $b$ is a constant. Call its size $s$.

Assume $p:=\ch(\F)$ is zero (or larger than $\max_{i,j}\{e_{i,j}\}$). Using the {\em duality} trick (see \cite[Theorem 2.1]{SSS12}), there exists another representation of $C$ as $C' := \sum_{i=1}^{k'} \prod_{j=1}^{n} g_{i,j}(x_j)$ of size $s^{O(b)}$. Rewrite this, using the obvious Hadamard algebra $\rmH_{k'}(\F)$, as - $C' = c^T\cdot D$, where $D = G_1(x_1) \star\cdots\star G_n(x_n) \in \rmH_{k'}(\F)[\term{x}]$. Trivially, the monomial-weight of each $G_j$ is bounded by $1$. Thus, by invoking Theorem \ref{thm:low-block-support} (\& the `key argument' in Lemma \ref{lem:sparse-poly}) we can shift $D$, in time $s^{O(\log k')}$, such that it becomes $O(\log k')$-concentrated. On top of the shift, the usual sparse PIT gives a hitting-set for $C$ in time $s^{O(\log s)}$. \qed


\subsection{Proof of Corollary \ref{cor:power-set-multilinear} }

Suppose we are given a blackbox access to the formula $C = \sum_{i=1}^k \prod_{j=1}^d f_{i,j}(\term{x}_{X_j})^{e_{i,j}}$, where $f_{i,j}$ is a sum of univariate polynomials in $\F[\term{x}_{X_j}]$, $e_{i,j} \in \N$, and $X_1\sqcup\cdots\sqcup X_d$ partitions $[n]$. Let the formula size be $s$.

Assume $\ch(\F)$ is zero (or larger than $\max_{i,j}\{e_{i,j}\}$). Using the {\em duality} trick (see \cite[Theorem 2.1]{SSS12}), there exists another representation of $f_{i,j}(\term{x}_{X_j})^{e_{i,j}}$ as $F_{i,j} := \sum_{p=1}^{k_{i,j}} \prod_{q\in X_j} g_{i,j,p,q}(x_q)$ of size $s^{O(1)}$. Trivially, the monomial-weight of each $g_{i,j,p,q}$ is bounded by $1$.
Overall, we can represent $C$ now as $C' := \sum_{i=1}^k \prod_{j=1}^d F_{i,j}$, which is a set-depth-$6$ formula. 
Recall the inductive proof of Theorem \ref{thm:low-support} on $C'$. It will have $H=3$ inductive steps. The crucial observation is that in the base case (dealing with sparse polynomials) we can use a better bound $\ell' = 2$ in Lemma \ref{lem:sparse-poly}, as $\mu(g_{i,j,p,q})\le1$. This leads us to an improvement on Theorem \ref{thm:low-support} -  we construct $\tau_0$ such that $\tau_0\circ D_0$ is $O(\log^2 s)$-concentrated over $\cR'_1[t_0]$, in time polynomial in $s^{\log^2 s}$. Again, on top of the shift, the usual sparse PIT gives a hitting-set for $C$ in time $s^{O(\log^2 s)}$. \qed

\section{Conclusion}

We have identified a natural phenomena - low-support rank-concentration -  in constant-depth formulas, that is directly useful in their blackbox PIT (up to quasi-polynomial time). In this work we gave a proof for the interesting special case of set-depth-$\D$ formulas. More work is needed to prove such rank-concentration in full generality.  Next, it would be interesting to prove rank-concentration for depth-$3$ formulas. 
Another direction is to improve this proof technique to give polynomial-time hitting-sets for set-depth-$\D$ formulas.

\section*{Acknowledgments} 

This work was initiated when MA and NS visited Max Planck Institute for Informatics, and would like
to thank the institute for its generous hospitality. The travel of MA was funded by Humboldt Forschungspreis, and that of NS by MPII. CS and NS would like to thank Hausdorff Center for Mathematics (Bonn) for the generous support during the research work. Additionally, CS is supported by the IMPECS fellowship. 
 
\bibliographystyle{amsalpha}
\bibliography{refs}

\appendix
\section{Diagonal circuits: The spirit of the argument} \label{sec:appendix_diag}
A circuit $C = \sum_{i=1}^{k}{f_i^d}$ is a diagonal circuit if $f_i$ is a linear polynomial in $n$ variables, $\term{x}$. \footnote{A lemma by Ellison \cite{E69} states that every $n$-variate polynomial of degree $d$ over $\C$ has a diagonal circuit representation although $k$ can be exponentially large.} We can associate a formula over a Hadamard algebra with $C$, namely
\begin{equation*}
D(\term{x}) := F^d \hspace{0.1in} \hspace{0.1in} \text{over } \rmH_{k}(\F), 
\end{equation*}
where $F = z_0 + z_1 x_1 + \ldots + z_n x_n$, every $z_j \in \F^k$ and $F$ restricted to the $i$-th coordinate of the vectors $z_0, \ldots, z_n$ is the linear polynomial $f_i$. Clearly, $C = (1, \hspace{0.01in} 1, \ldots ,1) \cdot D(\term{x})$, where $\cdot$ is the usual matrix product. Assume that $\text{char}(\F) = 0$ or $> d$.

Consider shifting every $x_j$ by a formal variable $t_j$, i.e. $x_j \mapsto x_j + t_j$. Then,
\begin{equation*}
D(\term{x} + \term{t}) = F(\term{x} + \term{t})^d = D(\term{t}) \star (1 + z'_1x_1 + \ldots + z'_nx_n)^d =:  D(\term{t}) \star D'(\term{x}),
\end{equation*}
where $z'_j = D(\term{t})^{-1} z_j$. We have stated before (in Section \ref{sec:intro}) that variables would be ultimately shifted by field constants. Here is a way to set $t_j$ a field constant: To ensure that $D(\term{t})^{-1}$ makes sense when $t_j$'s are set to constants, we map $t_j \mapsto y^j$ where $y$ is a fresh variable and then set $y$ to an $\alpha\in \F$ such that $\alpha$ is not a root of any of the polynomials $f_i(y, y^2, \ldots, y^n)$, $1 \leq i \leq k$. With this setting, we can safely assume that $D(\term{t})$ and $z'_1, \ldots, z'_n \in \rmH_{k}(\F)$.

Clearly, $C(\term{x} + \term{t}) = (1, \hspace{0.01in} 1, \ldots, 1) \cdot D(\term{x} + \term{t})$ is zero if and only if $C=0$. We would like to show that for $\ell = \lceil \log k \rceil$, $C(\term{x} + \term{t})$ is $\ell$-concentrated over $\F$. The coefficient of a monomial $\term{x}^{e} = \prod_{j \in [n]}{x_j^{e_j}}$ in $D(\term{x} + \term{t})$ is $D(\term{t}) \star \coef(e)(D') = {d \choose e}D(\term{t}) \star \prod_{j \in [n]}{{z'_j}^{e_j}} = {d \choose e}D(\term{t}) \star \term{z}'^{e}$, where ${d \choose e} = {d \choose {e_1, \ldots, e_n}}$. For a moment, treat $\term{z}'^{e}$ as a `monomial' in $z_1, \ldots, z_n$. List down all monomials in $z_1, \ldots, z_n$ with degree bounded by $d$ in degree-lexicographic order. The idea is to form a basis of $\lrsp_{\F}\{ \coef(e)(D') \,|\, e\in\N^n\}$ by picking terms $\term{z}'^{e}$, the coefficient of $\term{x}^e$ in $D'$ (upto scaling by ${d \choose e}$), from the ordered list. We pick a term $z_{j_1}^{e_1} \ldots z_{j_m}^{e_m}$ ($e_j > 0$) from the ordered list if it is not in the span of the already picked terms. The claim is, if $z_{j_1}^{e_1} \ldots z_{j_m}^{e_m}$ ($e_j > 0$) is picked then so are the terms $\prod_{r \in S}{z_{j_r}}$, for every set $S \subseteq [m]$ - this follows easily from the degree-lexicographic ordering of the list. This implies that $m < \lceil \log k \rceil = \ell$, as dimension of $\lrsp_{\F}\{ \coef(e)(D') \,|\, e\in\N^n\}$ is bounded by $k$ and there are $2^m$ such terms $\prod_{r \in S}{z_{j_r}}$. Therefore, $D'(\term{x})$ is $\ell$-concentrated over $\rmH_{k}(\F)$ which implies that $D(\term{x} + \term{t}) = D(\term{t}) \star D'(\term{x})$ is $\ell$-concentrated over $\rmH_{k}(\F)$. Since, $C(\term{x} + \term{t}) = (1, \hspace{0.01in} 1, \ldots, 1) \cdot D(\term{x} + \term{t})$, $C(\term{x} + \term{t})$ is also $\ell$-concentrated over $\F$.

Thus, by shifting $x_j \mapsto x_j + \alpha^j$, where $\alpha \in \F$ is such that none of the $f_i(\alpha, \alpha^2, \ldots, \alpha^n)$ is zero, we are guaranteed that the shifted diagonal circuit satisfies $\lceil \log k \rceil$-concentration. Such an $\alpha$ is always present among a set of $kn + 1$ distinct elements of $\F$. A quasi-polynomial hitting set generator for $C(\term{x})$ ensues immediately (as sketched in Section \ref{sec:intro}).   

\section{Missing proofs of Section \ref{sec:basics}} \label{sec:appendix_basics}
\subsection{Proof of Lemma \ref{lem:rec-struc}}
\begin{proof}
Recall that $f_j(\term{x}_{X_j}) = (f_{1,j}(\term{x}_{X_j}),  \ldots, f_{k,j}(\term{x}_{X_j}))^T$, where every $f_{i,j}(\term{x}_{X_j})$ is a set-height-$(H-h-1)$ formula over $\mathcal{R}_h$. The proof is by induction on height $(H-h-1)$ of $f_j(\term{x}_{X_j})$ (in other words, {\em reverse} induction on $h$). \\

\noindent \textit{Base case ($h+1\ge H-1$):} The base case is when $H-h-1 = 1$ or $0$, i.e.~$f_{i,j}(\term{x}_{X_j})$'s are sparse polynomials or linear polynomials depending on whether $\D$ is even or odd, repectively. In this case, $f_j(\term{x}_{X_j})$ is a set-height-$(H-h-1)$ formula over $\mathcal{R}_{h+1}$. Also, the sparsity parameter $\lambda$ remains the same by its definition. Hence, $f_j(\term{x}_{X_j}) \in \mathcal{C}_{h+1}(k,d, \lambda,  \term{x}_{X_j})$. (Here we do not care about the partition.)\\

\noindent \textit{Inductive step ($h+2$ to $h+1$):} The crucial property to note here is that the formulas $f_{i,j}(\term{x}_{X_j})$'s appear as sub-formulas of $C_h$ at depth-$3$ (Equation \ref{eqn:mainformula}). Therefore, the \emph{corresponding} $\Pi$-layers of $f_{1,j}(\term{x}_{X_j}), \ldots, f_{k,j}(\term{x}_{X_j})$ respect the \emph{same} partitions of $\term{x}_{X_j}$. In particular, we can express every $f_{i, j}(\term{x}_{X_j})$ as,
\begin{equation*}
f_{i, j}(\term{x}_{X_j}) = \sum_{p=1}^{k}{b_{i,j,p} \cdot \prod_{q=1}^{d}{g_{i,j,p,q}(\term{x}_{Y_{j,q}})}},
\end{equation*}
where $b_{i,j,p} \in \mathcal{R}_h$, $g_{i,j,p,q}(\term{x}_{Y_{j,q}})$ is a set-height-$(H-h-2)$ formula over $\mathcal{R}_h$, and the first $\Pi$-layer of all $f_{i,j}(\term{x}_{X_j})$, for $1 \leq i \leq k$, respect the same partition $\mathcal{P}_h(2, X_j)$. In other words, $Y_{j,q}$'s partition $X_j$ as do $X_{2,q}\cap X_j$. (Note: With $j$ fixed, here $X_{2,q} \cap X_j$ are the only relevant variable indices.) Hence, 
\begin{equation} \label{eqn:rec}
f_j(\term{x}_{X_j}) = \sum_{p=1}^{k}{b_{j, p} \cdot \prod_{q=1}^{d}{g_{j,p,q}(\term{x}_{Y_{j,q}})}}, 
\end{equation}
where $b_{j, p} = (b_{1,j,p}, \cdots, b_{k,j,p})^T \in \mathcal{R}_{h+1}$ and $g_{j,p,q}(\term{x}_{Y_{j,q}}) = (g_{1,j,p,q}(\term{x}_{Y_{j,q}}), \ldots, g_{k,j,p,q}(\term{x}_{Y_{j,q}}))^T$ $\in \mathcal{R}_{h+1}[\term{x}_{Y_{j,q}}]$. 

In order to apply induction, we make a comparison between $f_{i,j}(\term{x}_{X_j})$ and $g_{i,j,p,q}(\term{x}_{Y_{j,q}})$ (and between $f_j(\term{x}_{X_j})$ and $g_{j,p,q}(\term{x}_{Y_{j,q}})$). Just like $f_{i,j}(\term{x}_{X_j})$ is a set-height-$(H-h-1)$ formula over $\mathcal{R}_h$ occurring as a sub-formula at depth-$3$ of the formula $C_h$, $g_{i,j,p,q}(\term{x}_{Y_{j,q}})$ is a set-height-$(H-h-2)$ formula over $\mathcal{R}_h$ occurring as a sub-formula at depth-$5$ of the formula $C_h$. Hence, by induction, $g_{j,p,q}(\term{x}_{Y_{j,q}})$ is a set-height-($H-h-2$) formula in $\mathcal{R}_{h+1}[\term{x}_{Y_{j,q}}]$ with $\Sigma$-fanin $k$, $\Pi$-fanin $d$ and sparsity parameter $\lambda$ i.e., $g_{j,p,q}(\term{x}_{Y_{j,q}}) \in \mathcal{C}_{h+2}(k,d, \lambda,  \term{x}_{Y_{j,q}})$, such that every $h'$-th $\Pi$-layer of $g_{j,p,q}(\term{x}_{Y_{j,q}})$ respects the partition $\mathcal{P}_h(h'+2, Y_{j,q})$. Since $g_{j,p,q}(\term{x}_{Y_{j,q}})$ has only variables $\term{x}_{Y_{j,q}}$ and $Y_{j,q} \subseteq X_j$, we can also say that every $h'$-th $\Pi$-layer of $g_{j,p,q}(\term{x}_{Y_{j,q}})$ respects the partition $\mathcal{P}_h(h'+2, X_{j})$. The $h'$-th $\Pi$-layers of the $g_{j,p,q}(\term{x}_{Y_{j,q}})$'s (for $1 \leq q \leq d$) correspond to the $(h'+1)$-th $\Pi$-layer of $f_j(\term{x}_{X_j})$. Hence, by Equation \ref{eqn:rec}, we infer that every $h'$-th $\Pi$-layer of $f_j(\term{x}_{X_j})$ respects the partition $\mathcal{P}_h(h'+1, X_j)$. Note that the $\Sigma$-fanin, $\Pi$-fanin and the sparsity parameter remain $k, d$ and $\lambda$, respectively. This proves the claim.
\end{proof}

\section{Missing proofs of Section \ref{sec:low-blk-supp}} \label{sec:appendix_lbs}
\subsection{Proof of Lemma \ref{lem-struc-eqn}}
\begin{proof}
Consider a column $u\in\cS$ of $Z'$; it is $z'_u$. Now
\begin{eqnarray*}
z'_u 
&=& f(\term{t})^{-1}\star \sum_{v\in S} z_{v} {v\choose u} t^{v-u}\quad \text{ [by Equation \ref{eqn-zi-prime}]}\\
&=& f(\term{t})^{-1}\star \sum_{v\in \cS} z_v \cdot t^v \cdot {v\choose u} \cdot t^{-u} \\
&=& f(\term{t})^{-1}\star Z \cdot (u\text{-th column of } N_\cS T N_\cS^{-1} ).
\end{eqnarray*}
Running over all $u\in\cS$ gives us the result.
\end{proof}

\subsection{Proof of Lemma \ref{lem-Teqn-mod}}
\begin{proof}
Lemma \ref{lem-struc-eqn} gives $Z'_\cS = f(\term{t})^{-1}\star Z N_\cS T_{\cS,\cS} N_\cS^{-1}$. Rewrite it as,
$$f(\term{t})^{-1}\star Z = Z'_\cS N_\cS T' N_\cS^{-1}.$$
Going modulo the subspace $\lrsp_{\F(\term{t})}\{z'_0\}$ kills the $0$-th column of $Z'_\cS$ and yields,
$$f(\term{t})^{-1}\star Z \equiv Z'_{\cS^*} N_{\cS^*} T'_{\cS^*,\cS} N_\cS^{-1} \pmod{z'_0}.$$ 

For the second part we exploit the independence of $T'_{\cS^*,\cS}$ from $Z$ and the Hadamard algebra. Formally, fix a large enough $\widetilde{\k}$,  say $|\cS|$, and the Hadamard algebra $\rmH_{\widetilde{\k}}(\F)$. Let $e\in \cS$. Fix $\widetilde{Z}$ as: Its $e$-th column is $0$ and the rest are linearly independent modulo $1$ (note: $1=\widetilde{z}'_0$). For this `generic' setting we still have the equation,
$\widetilde{f}(\term{t})^{-1}\star \widetilde{Z} \equiv \widetilde{Z}'_{\cS^*} N_{\cS^*} T'_{\cS^*,\cS} N_\cS^{-1} \pmod{\widetilde{z}'_0}$. Implying,
$$\widetilde{f}(\term{t})^{-1}\star \widetilde{Z}_{\cS\setminus\{e\}} \equiv \widetilde{Z}'_{\cS^*} N_{\cS^*} T'_{\cS^*, \cS\setminus\{e\}} N_{\cS\setminus\{e\}}^{-1} \pmod{\widetilde{z}'_0}.$$
Since the LHS is a matrix of rank $|\cS|-1$, we deduce that $T'_{\cS^*, \cS\setminus\{e\}}$ is invertible. In other words, $T'_{\cS^*,\cS}$ is strongly full.
\end{proof}

\subsection{Proof of Lemma \ref{lem-Teqn-depth3}}
\begin{proof}
For $i\in[\ell]$, we can apply Lemma \ref{lem-Teqn-mod} to $f_i$ and get,
\begin{equation}\label{eqn-Teqn-i}
f_i(\term{t})^{-1}\star Z_i \equiv  Z'_i N_{\cS^*_i} T'_i N_{\cS_i}^{-1} \pmod{1}
\end{equation}
where the $1$ is the unity, the all one vector, in $\rmH_{\k}(\F)$. Denote the $u_i$-th column of the matrix on the RHS, of the above congruence, by $C_{i,u_i}$.
 
Consider a column $u\in \cS$ of $Z$; it is $z_u$. Now
\begin{eqnarray*}
D(\term{t})^{-1}\star z_u 
&=& \prod_{i\in[\ell]} f_i(\term{t})^{-1}\star z_{i,u_i} \\
&=& \prod_{i\in[\ell]} \( \alpha_i + C_{i,u_i} \) \quad \text{ [for some $\alpha_i\in\F(\term{t})$ by Equation \ref{eqn-Teqn-i}]}\\
&\equiv & \prod_{i\in[\ell]} C_{i,u_i} \pmod{\cV_{\ell}(D')} \quad\text{ [$\because$ the product of $\ell$ or less $C_{i,u_i}$ vanishes]} 
\end{eqnarray*}
Running over all $u\in \cS$ gives us,
\begin{eqnarray*}
D(\term{t})^{-1}\star Z 
&\equiv & \circledast_{i\in[\ell]} \(Z'_i N_{\cS^*_i} T'_i N_{\cS_i}^{-1}\) \\
&\equiv & \(\circledast_{i\in[\ell]} Z'_i\)\cdot \otimes_{i\in[\ell]} \(N_{\cS^*_i} T'_i N_{\cS_i}^{-1}\) \quad \text{ [by Lemma \ref{lem-matrices}-(4)]}\\
&\equiv & Z' \cdot N_{\cS'} \cdot T' \cdot N_\cS^{-1} \pmod{\cV_{\ell}(D')} \quad \text{ [by Lemma \ref{lem-matrices}-(1)]}
\end{eqnarray*}
\end{proof}

\subsection{Proof of Theorem \ref{thm-sub-prop}}
\begin{proof}
We know that $T' = \otimes_{i\in[\ell]}T'_i$, where each $T'_i\in (\cS^*_i \times \cS_i \rightarrow\F)$ is strongly full (Lemma \ref{lem-Teqn-mod} for $f_i$). Thus, we can apply invertible row operations $E_i\in (\cS^*_i \times\cS^*_i  \rightarrow\F)$ such that $E_iT'_i$ has a $|\cS^*_i|$-sized identity submatrix, and another column that has only nonzero entries.

Since, from now on, we are not going to use the properties of the index sets $\cS^*_i, \cS_i$, we replace them by a more readable identification: Define, for $i\in[\ell]$, $n_i:=|\cS^*_i|>0$ and identify $\cS^*_i$ (resp.~$\cS_i$) with $U_i:=[n_i]$ (resp.~$W_i:=[0..n_i]$). Let $U := \times_{i\in[\ell]} U_i$ and $W := \times_{i\in[\ell]} W_i$. Wlog we keep the following setting: For all $i\in[\ell]$,
\begin{enumerate}
\item 
$(T'_i)_{U_i,U_i} = I_{n_i}$ [by Lemma \ref{lem-matrices}-(1), and taking $E_iT'_i$ to be our new $T'_i$], and
\item 
the column $(T'_i)_{U_i,0}$ is zero free.
\end{enumerate}
Define an {\em indicator} function (note: $\delta(\cdot)$ equals $1$, if the boolean condition is true, else $0$)
$$\ve: \N_{>0}\times\N \rightarrow \{0,1\};\, (u,w) \mapsto \delta\( (w=0) \vee (w\ne0 \wedge w=u) \).$$
Extend it to $\N_{>0}^\ell \times \N^\ell$ by defining $\ve:(u,w) \mapsto \prod_{r\in[\ell]} \ve(u_r,w_r)$.

Note that the $(u,w)$-th entry in $T'_i$ is nonzero iff $\ve(u,w)=1$. Thus, $\ve$ {\em exactly indicates the non-zeroness in} $T'_i$.

Similarly, by tensoring, the $(u,w)$-th entry in $T'\in(U \times W \rightarrow \F )$ is nonzero iff $\ve(u,w)=1$. Thus, $\ve$ {\em exactly indicates the non-zeroness in} $T'$.

We will build $\cC$ incrementally, starting with $\cC=\emptyset$.  During this build up we might apply row permutations $R$ on $T'$.
 
Consider a column $u$, $u\in U\subset W$, of $T'$. This column has exactly one nonzero entry; appearing at the row indexed by $u\in U$. Put all these unmarked columns $u$ in $\cC$, and collect the marked ones in $\cM_1$.

If $\cM_1=\emptyset$ then we already have $|\cC|=|U|$ and we are done (infact, $T'_{U,\cC}$ is identity). So assume $|\cM_1|=:m_1\in[\k]$ and define $m_2:=\k-m_1<\k$. Let the other marked columns be $\cM_2:=\cM\setminus\cM_1$; they lie in $W\setminus U$ and are $m_2$ many.    

Consider the unmarked columns in $W\setminus U$; collect them in $\cL:=W\setminus(U\cup\cM_2)$. We will now focus on the submatrix $T'_{\cM_1, W\setminus U}=:T'_1$. Note that its column-indices are $\ell$-tuples with at least one zero.

\begin{claim}\label{clm-T-pr}
There exists a row-permutation $R_1\in\F^{m_1\times m_1}$, and $m_1$ unmarked columns $\cC_1\subseteq\cL$ such that: $(R_1T'_1)_{\cM_1,\cC_1}$ is a lower-triangular $m_1\times m_1$ matrix with $w$-th ($w\in\cC_1$) diagonal entry being nonzero.
\end{claim}
\claimproof{clm-T-pr}{
We will again build $\cC_1$ incrementally, starting from $\emptyset$.

Recall that each row of $T'_1$ is indexed by an $\ell$-tuple $u$ in $U$. For $i\in[\ell]$ we denote the $i$-th coordinate in $u$ by $u(i)$, and for an $I\subseteq[\ell]$, $u(I)$ denotes the ordered set $\{u(i)|i\in I\}$. For $w\in W$, define the {\em support} $\rmS(w) := \{i\in[\ell] \,|\, w(i)\ne0\}$.
We want to permute the rows so that the coordinates of the row-indices appear in a {\em decreasing order of frequency}. Formally, pick $R_1\in\F^{m_1\times m_1}$ to reorder the rows of $T'_1$ as $\cM_1=(u_1,\ldots,u_{m_1})$ such that:
\begin{itemize}
\item 
The ordered list $u_1(1),\ldots,u_{m_1}(1)$ has repetitions only in contiguous locations and the frequencies are non-increasing. In equation terms: The list has some $r$ distinct elements $\alpha_1,\ldots,\alpha_r\in U_1$ with respective frequencies $i_1\ge\cdots\ge i_r$ (summing to $m_1$), and they appear as $\alpha_1 (i_1\text{ times}),\ldots,\alpha_r (i_r\text{ times})$. 
\item 
The ordered list $(u_1(1),u_1(2)),\ldots,(u_{m_1}(1),u_{m_1}(2))$ has repetitions only in contiguous locations and the frequencies are non-increasing.
\item
The same as above holds for $3$-tuples, $4$-tuples,$\ldots$,$\ell$-tuples. 
\end{itemize}
We now describe an iterative process to build $\cC_1$ one element at a time. In the $i$-th iteration, $i\in[m_1]$, we will add an unmarked, {\em unpicked} column $w_i\in\cL$ to $\cC_1$. The process maintains the invariant: $(R_1T'_1)_{\cM_1,\cC_1}$ is a {\em lower-triangular} matrix. 

\smallskip\noindent
{\em Iteration $i=1$} - The row $u_1$ of $T'_1$ has exactly $2^{\ell}-1$ nonzero columns. (Why?~Zero-out at least one coordinate of $u_1$.) Since $2^{\ell}-1 \ge \k > |\cM_2|$ we can pick a 
column $w_1\in\cL$ such that $\ve(u_1,w_1) \ne 0$, thus $(T'_1)_{u_1,w_1} \ne 0$. Add $w_1$ to $\cC_1$.

\smallskip\noindent
{\em Iteration $i\ge2$} - Consider the list $u_1,\ldots,u_i$. We claim that there are positions $I\subset[\ell]$, $|I|\le\ce{\lg i}$, such that $u_i(I)$ is not contained in any of the previous sets in the list. The proof is by {\em binary-search} in the list. Start with $I=\emptyset$. Pick the least $j_1\in[\ell]$ such that $u_1(j_1),\ldots,u_{i}(j_1)$ are not all the same; add $j_1$ to $I$. By the ordering on $u$'s the frequency $\mu_1$ of $u_{i}(j_1)$ is at most $i/2$. If it is one then we stop with this $I$, otherwise we zoom-in on the `halved' list $u_{i-\mu_1+1},\ldots,u_{i}$. Again we pick the least $j_2\in[j_1+1,\ell]$ such that $u_{i-\mu_1+1}(j_2),\ldots,u_{i}(j_2)$ are not all the same; add $j_2$ to $I$. This leads to a further halving of the list, and so on. Finally, we do have our positions $I$, $|I|\le\ce{\lg i}$, such that $u_i(I)$ appears for the first time in $u_i$. 

We deduce that each column $w$ of $T'_1$, with $I\subseteq \rmS(w) \subsetneq [\ell]$ and $w(\rmS(w)) = u_i(\rmS(w))$, has the first nonzero entry at the $u_i$-th row. (Why?~Consider $\ve( u_j, w ) = \ve( u_j(\rmS(w)), w(\rmS(w)) ) = \ve( u_j(\rmS(w)), u_i(\rmS(w)) )$.) The number of such columns $w$, that are unmarked and unpicked, is at least $(2^{\ell-|I|}-1)-m_2-(i-1) \ge$ $2^{\ell-|I|}-\k \ge $  $2^{\ell-\ce{\lg i}}-\k \ge$ $2^{\ell-\ce{\lg \k}}-\k =$ $2^{\ce{\lg \k}+1}-\k > 0$. So we can pick such a column, say, $w_i \in \cL\setminus\cC_1$ and add to $\cC_1$. 

Note that the square submatrix of $T'_1$ thus far, $(R_1T'_1)_{\{u_1,\ldots,u_i\},\cC_1}$ is lower-triangular with a nonzero diagonal.

\smallskip\noindent
{\em After the iteration $i=m_1$} - The square matrix $(R_1T'_1)_{\cM_1,\cC_1}$ is lower-triangular with a nonzero diagonal. 

This finishes the claim.
}

Since $R_1$ permutes the rows of $T'_1$, its action can be lifted to the rows of $T'$; call this action $R$. Also, append $\cC_1$ to the current $\cC$ (making its size $|U|$). Define $\ol{\cM}_1 := U\setminus\cM_1$ and $\ol{\cC}_1 := \cC\setminus\cC_1$. Consider the square matrix $(RT')_{U,\cC}$. It looks like,
$$
\left[\begin{array}{c|c}
(RT')_{\ol{\cM}_1,\ol{\cC}_1} & (RT')_{\ol{\cM}_1,\cC_1} \\ \hline
(RT')_{\cM_1,\ol{\cC}_1} & (RT')_{\cM_1,\cC_1}
\end{array}\right] 
= 
\left[\begin{array}{c|c}
I_{\ol{\cM}_1,\ol{\cC}_1} & (RT')_{\ol{\cM}_1,\cC_1} \\ \hline
0_{\cM_1,\ol{\cC}_1}         & (R_1T'_1)_{\cM_1,\cC_1}
\end{array}\right].
$$
Clearly, its determinant equals $|(R_1T'_1)_{\cM_1,\cC_1}| \ne0$. Thus, $|T'_{U,\cC}| \ne0$ and we are done.
\end{proof}

\subsection{Proof of Lemma \ref{lem-tna-det}}
\begin{proof}
Let $a$ be the $v$-th column of $A$. Let $a'\in\F^{|\cM|}$ be the vector having the entries of $a$ appearing at the rows $\cM$. Consider $(T' N_\cS^{-1})\cdot a$. By the property of $a$ we can write, 
\begin{eqnarray*}
(T' N_\cS^{-1})a  
&=&  (T' N_\cS^{-1})_{\cS', v} + (T' N_\cS^{-1})_{\cS',\cM}\cdot a' \\
&=&  T'_{\cS', v}\cdot t^{-v} + (T' N_\cS^{-1})_{\cS',\cM}\cdot a'.
\end{eqnarray*}
Thus, the $v$-th column of $A$ has the leading monomial $t^{-v}$ which `contributes' the vector $T'_{\cS', v}$. Going over the columns $a$, running $v\in\cC$, by the column-linearity of determinant and the multiplicativity of the inverse-monomial ordering, we deduce that the largest possible (inverse-monomial) term in the expression $|T' N_\cS^{-1} A|$ is:
$$|T'_{\cS', \cC}|\cdot t^{- \sum_{v\in\cC} v}.$$ 
We know this is nonzero, by the property of $\cC$, thus it is {\em indeed} the leading term. In particular, $|T' N_\cS^{-1} A|\ne0$.
\end{proof}

\section{Missing proofs of Section \ref{sec:low-supp}} \label{sec:appendix_ls}
\subsection{Proof of Lemma \ref{lem:sparse-poly}}
\begin{proof}
If $2\ce{\log_2 (\k\cdot\rms(f))} \ge \mu(f)$ then $\ell' = 1 + \mu(f)$. In this case trivially, for any shift $\sigma$, $\sigma(f)$ is $\ell'$-concentrated over $\rmH_{\k}(\F(t))$. So, from now on we assume $2\ce{\log_2 (\k\cdot\rms(f))}$ $< \mu(f)$, thus $\ell' = 1 + 2\ce{\log_2 (\k\cdot\rms(f))}$.

Define $\cR := \rmH_{\rms(f)}(\rmH_{\k}(\F))$. Let $f =: \sum_{e\in\rmS(f)} z_e x^e$. Define a column vector $D \in (\rmS(f)\times [1]\rightarrow \rmH_{\k}(\F[\term{x}]))$ with $e$-th entry being $z_e x^e$; $D$ can be seen as a polynomial over $\cR$. Rewrite $D$ as a product of univariate polynomials over $\cR$ as:
$$D(\term{x}) \,=\, g_1(x_1) \star\cdots\star g_n(x_n).$$
Clearly, each $g_i$ has degree, hence sparsity, bounded by $\delta$, and can be seen as an element in $\rmH_{\k\cdot\rms(f)}(\F)[x_i]$. 

For any $X\subseteq[n]$ of size $\ell'$, define $D_X(\term{x}) \,:=\, \prod_{i\in X} g_i(x_i)$.
Recalling Theorem \ref{thm:low-block-support} we can construct a shift $\sigma$ for $D_X$, such that $\sigma\circ D_X$ is $\ell'$-concentrated, in time polynomial in $(\delta+n+\ell')^{\ell'}$. Using induction on the number of variables, it is easy to see that if $\sigma\circ D_X$ is $\ell'$-concentrated ($\forall X\in {[n]\choose \ell'}$) then so is $\sigma\circ D$. The {\em key argument} is: Since the constant coefficient in each $g'_i$ (i.e.~shift-\&-normalized $g_i$) is one, deduce that the coefficient of any term in $D'$ (i.e.~shift-\&-normalized $D$) of block-weight $\le\ell'$ is produced by the product of some $\le\ell'$ $g'_i$'s, so this case is covered by some $X\in {[n]\choose \ell'}$. Also, deduce that the coefficient of any term in $D'$ of block-weight $>\ell'$ can be inductively written down as a linear combination of $\{ \coef(e)(D') \,|\, e\in\N^n, \rms(e)<\ell' \}$. Finally, $\sigma\circ D$ inherits this concentration property from $D'$.

Recall $f \,=\, 1^T\cdot D$, where $1$ is the unity in $\cR = \rmH_{\rms(f)}(\rmH_{\k}(\F))$. Thus, from the $\ell'$-concentration of $\sigma\circ D$ (over $\cR$), we can deduce the $\ell'$-concentration of $\sigma\circ f$ (over $\rmH_{\k}(\F)$). This completes the construction of $\sigma$.
\end{proof}

\subsection{Proof of Lemma \ref{lem:hitting-set-non-sparse-poly}}
\begin{proof}
View $f$ as a vector with $\k$ coordinates; each entry is in $\F[\term{x}]\setminus\{0\}$. Call the $i$-th entry $f_i$. Clearly, $f_i$ has variables (resp.~degree) at most $n$ (resp.~$\delta$). Also, by the concentration property there exists $e_i\in\N^n$, with $\rms(e_i)\le\ell'$, such that $\coef(e_i)(f_i) \ne0$. 

For $X\subseteq[n]$ of size at most $\ell'$, define $\sigma_{X}:x_j\mapsto (x_j \text{ if } j\in X, \text{ else }0)$
for all $j\in[n]$. Clearly, $\sigma_X\circ f_i$ is only $\ell'$ variate, thus it has sparsity $(\delta+\ell')^{O(\ell')}$. By the assumption on $f_i$ we know that $X_i := \rmS(e_i)$ is of size at most $\ell'$, and $\sigma_{X_i}\circ f_i \ne0$. Using standard sparse PIT methods (see \cite{BHLV09}), we can construct a hitting-set for $\sigma_{X_i}\circ f_i$ in time $(\delta+\ell')^{O(\ell')}$. Varying over all subsets $X\subseteq[n]$ of size at most $\ell'$, we get a hitting-set for $f_i$ in time $(\delta+n+\ell')^{O(\ell')}$.
For convenience, denote this hitting-set as a set of evaluation-maps $\{\sigma_{i,1},\ldots,\sigma_{i,r}\}$; each map is from $\term{x}$ to $\F$ and we write $\sigma_{i,j}\circ f_i$ to mean $f_i(\sigma_{i,j}(\term{x}))$. Overall we are ensured the existence of a $j$, for a given $i$, such that $\sigma_{i,j}\circ f_i \ne0$. We will now show how to combine all these into a single map.

Pick distinct $\k r$ elements $\beta_{1,1},\ldots, \beta_{\k,r} \in\F$. Consider the univariate polynomial $g(u) := \prod_{i\in[\k], j\in[r]} (u-\beta_{i,j})$. Define $g_{i,j}(u) := g(u)/(u-\beta_{i,j})$, for all $i,j$. Consider an evaluation map from $\F[\term{x}]$ to $\F[u,v]$ -  $\sigma := v\cdot\sum_{i\in[\k], j\in[r]} g_{i,j}(u)\cdot \sigma_{i,j}$. We claim that, for all $i\in[\k]$, $\sigma\circ f_i \ne0$. To see this, note that there is some $j\in[r]$ for which $\sigma_{i,j}\circ f_i \ne0$. Further, let $f'_i$ be a homogeneous part of $f_i$, say of degree $\delta_i$, such that $\sigma_{i,j}\circ f'_i \ne0$. Consider the partial evaluation $(\sigma\circ f_i)(\beta_{i,j},v) = f_i(v\cdot g_{i,j}(\beta_{i,j})\cdot\sigma_{i,j}(\term{x}))$. Here the coefficient of the monomial $v^{\delta_i}$ is $g_{i,j}(\beta_{i,j})^{\delta_i}\cdot (\sigma_{i,j}\circ f'_i) \ne0$. Consequently, $\sigma\circ f_i \ne0$.

Thus, for all $i\in[\k]$, $\sigma\circ f_i$ is a nonzero bivariate polynomial in $\F[u,v]$. Since its degree remains bounded by $\delta\cdot \k r$, we can again apply \cite{BHLV09} to replace $u,v$ by a hitting-set. Finally, we hit an $\bal\in\F^n$, in time polynomial in $\k (\delta+n+\ell')^{\ell'}$, such that for all $i\in[\k]$, $f_i(\bal) \ne0$. This finishes the proof.
\end{proof}

\end{document}